%% file: mmWave_SingleColumn.tex
\documentclass[journal,onecolumn,draftclsnofoot,11pt]{IEEEtran}
\input{input.tex}

\newcommand{\pinv}[1]{\ensuremath{#1^{\dagger}}} 	

\def\Nc{N_\mathrm{RF}}
\def\Fbb{\bF_\mathrm{BB}}
\def\Frf{\bF_\mathrm{RF}}
\def\Wbb{\bW_\mathrm{BB}}
\def\Wrf{\bW_\mathrm{RF}}

\newcommand{\sref}[1]{{Section}~\ref{#1}}
\usepackage{epstopdf}
\usepackage{enumerate}
\usepackage{algorithmicx}
\usepackage{algorithm}
\usepackage{amsmath}
\usepackage[noend]{algpseudocode}
\usepackage{float}
\usepackage{color}
\usepackage{makeidx}

\begin{document}
\IEEEoverridecommandlockouts
\title{Channel Estimation and Hybrid Precoding for Millimeter Wave Cellular Systems\thanks{This material is based upon work supported by the National Science Foundation under Grant No. 1218338 and 1319556.}}
\author{Ahmed Alkhateeb$^{\dag}$, Omar El Ayach$^{\dag}$, Geert Leus$^{\ddag}$, and Robert W. Heath Jr.$^{\dag}$\\
$^\dag$ The University of Texas at Austin, Email: $\{$aalkhateeb, oelayach,  rheath$\}$,@utexas.edu\\
$^\ddag$ Delft University of Technology, Email: {g.j.t.leus@tudelft.nl}}

\maketitle

\begin{abstract}
 Millimeter wave (mmWave) cellular systems will enable gigabit-per-second data rates thanks to the large bandwidth available at mmWave frequencies. To realize sufficient link margin, mmWave systems will employ directional beamforming with large antenna arrays at both the transmitter and receiver. Due to the high cost and power consumption of gigasample mixed-signal devices, mmWave precoding will likely be divided among the analog and digital domains. The large number of antennas and the presence of analog beamforming requires the development of mmWave-specific channel estimation and precoding algorithms. This paper develops an adaptive algorithm to estimate the mmWave channel parameters that exploits the poor scattering nature of the channel. To enable the efficient operation of this algorithm, a novel hierarchical multi-resolution codebook is designed to construct training beamforming vectors with different beamwidths. For single-path channels, an upper bound on the estimation error probability using the proposed algorithm is derived, and some insights into the efficient allocation of the training power among the adaptive stages of the algorithm are obtained. The adaptive channel estimation algorithm is then extended to the multi-path case relying on the sparse nature of the channel. Using the estimated channel, this paper proposes a new hybrid analog/digital precoding algorithm that overcomes the hardware constraints on the analog-only beamforming, and approaches the performance of digital solutions. Simulation results show that the proposed low-complexity channel estimation algorithm achieves comparable precoding gains compared to exhaustive channel training algorithms. The results also illustrate that the proposed channel estimation and precoding algorithms can approach the coverage probability achieved by perfect channel knowledge even in the presence of interference.
\end{abstract}

\section{Introduction} \label{sec:intro}

Millimeter wave (mmWave) communication  is a promising technology for future outdoor cellular systems~\cite{pi2011introduction,.11ad,.13c, Rapp5G}. Directional precoding with large antenna arrays appears to be inevitable to support longer outdoor links and to provide sufficient received signal power. Fortunately, large antenna arrays can be packed into small form factors at mmWave frequencies \cite{Antenna1, Antenna2}, making it feasible to realize the large arrays needed for high precoding gains. The high power consumption of mixed signal components, however, makes digital baseband precoding impossible \cite{pi2011introduction}. Moreover, the design of the precoding matrices is usually based on complete channel state information, which is difficult to achieve in mmWave due to the large number of antennas and the small signal-to-noise ratio (SNR) before beamforming. Because of the additional hardware constraints when compared with conventional microwave frequency multiple-input multiple-output (MIMO) systems, new channel estimation and precoding algorithms that are tailored to mmWave cellular systems must be developed.

To overcome the radio frequency (RF) hardware limitations, analog beamforming solutions were proposed in~\cite{Wang1, .13c, chen2011multi, Multilevel, Tsang}. The main idea is to control the phase of the signal transmitted by each antenna via a network of analog phase shifters. Several solutions, known as beam training algorithms, were proposed to iteratively design the analog beamforming coefficients in systems without channel knowledge at the transmitter. In \cite{Wang1,.13c, chen2011multi, Multilevel}, adaptive beamwidth beamforming algorithms and multi-stage codebooks were developed by which the transmitter and receiver jointly design their beamforming vectors. In \cite{Tsang}, multiple beams with unique signatures were simultaneously used to minimize the required beam training time. Despite the reduced complexity of \cite{Wang1, .13c, chen2011multi, Multilevel, Tsang, Zhang}, they generally share the disadvantage of converging towards only one communication beam. Hence, these techniques are not capable of achieving multiplexing gains by sending multiple parallel streams. Moreover, the performance of analog strategies such as those in \cite{Wang1,.13c, chen2011multi, Multilevel} is sub-optimal compared with digital precoding solutions due to (i) the constant amplitude constraint on the analog phase shifters, and (ii) the potentially low-resolution signal phase control.

To achieve larger precoding gains, and to enable precoding multiple data streams, \cite{Zhang, Venkat,ayach2013spatially,alkhateeb} propose to divide the precoding operations between the analog and digital domains. In \cite{Zhang}, the joint analog-digital precoder design problem was considered for both spatial diversity and multiplexing systems. First, optimal unconstrained RF pre-processing signal transformations followed by baseband precoding matrices were proposed, and then closed-form sub-optimal approximations when RF processing is constrained by variable phase-shifters were provided. In \cite{Venkat}, hybrid analog/digital precoding algorithms were developed to minimize the received signal's mean-squared error in the presence of interference when phase shifters with only quantized phases are available. The work in \cite{Zhang, Venkat}, however, was not specialized for mmWave systems, and did not account for mmWave channel characteristics. In \cite{ayach2013spatially}, the mmWave channel's sparse multi-path structure~\cite{rappaport2012cellular, murdock201238, zhang2010channel,Ben-Dor, sayeed2007maximizing}, and the algorithmic concept of basis pursuit, were leveraged in the design of low-complexity hybrid precoders that attempt to approach capacity assuming perfect channel knowledge is available to the receiver. In \cite{ayach2013spatially,alkhateeb}, the hybrid precoding design problem was considered in systems where the channel is partially known at the transmitter. While the developed hybrid precoding algorithms in \cite{Zhang, ayach2013spatially,alkhateeb} overcome the RF hardware limitations and can support the transmission of multiple streams, the realization of these gains require some knowledge about the channel at the transmitter prior to designing the precoding matrices. This motivates developing multi-path mmWave channel estimation algorithms, which enable hybrid precoding to approach the performance of the digital precoding algorithms.

In this paper, we develop low-complexity channel estimation and precoding algorithms for a mmWave system with large antenna arrays at both the base station (BS) and mobile station (MS). These algorithms account for practical assumptions on the mmWave hardware in which (i) the analog phase shifters have constant modulus and quantized phases, and (ii) the number of RF chains is limited, i.e., less than the number of antennas. The main contributions of the paper can be summarized as follows:

\begin{itemize}
\item{We propose a new formulation for the mmWave channel estimation problem. This formulation captures the sparse nature of the channel, and enables leveraging tools developed in the adaptive compressed sensing (CS) field to design efficient estimation algorithms for mmWave channels.}
\item{We design a novel multi-resolution codebook for the training precoders. The new codebook relies on joint analog/digital processing to generate beamforming vectors with different beamwidths, which is critical for proper operation of the adaptive channel estimation algorithms presented in the paper.}
\item{We design an adaptive CS based algorithm that efficiently estimates the parameters of mmWave channels with a small number of iterations, and with high success probability. The advantage of the proposed algorithm over prior beam training work appears in multi-path channels where our algorithm is able to estimate channel parameters. Hence, it enables multi-stream multiplexing in mmWave systems, while prior work \cite{tsang2011successive,Multilevel,chen2011multi,Wang1,van2002optimum} was limited to the single-beam training and transmission.}
\item{We analyze the performance of the proposed algorithm in single-path channels. We derive an upper bound on the error probability in estimating channel parameters, and find sufficient conditions on the total training power and its allocation over the adaptive stages of the algorithm to estimate the channel parameters with a certain bound on the maximum error probability.}
\item{We propose a new hybrid analog/digital precoding algorithm for mmWave channels. In the proposed algorithm, instead of designing the precoding vectors as linear combinations of the steering vectors of the known angles of arrival/departure as assumed in \cite{ayach2013spatially}, our design depends only on the quantized beamsteering directions to directly approximate the channel's dominant singular vectors. Hence, it implicitly considers the hardware limitations, and more easily generalizes to arbitrary antenna arrays.}
\item{We evaluate the performance of the proposed estimation algorithm by simulations in a mmWave cellular system setting, assuming that both the BS and MS adopt hybrid precoding algorithms.}
\end{itemize}

Simulation results indicate that the precoding gains given by the proposed channel estimation algorithm are close to that obtained when exhaustive search is used to design the precoding vectors. Multi-cell simulations show that the spectral efficiency and coverage probability achieved when hybrid precoding is used in conjunction with the proposed channel estimation strategy are comparable to that achieved when perfect channel knowledge and digital unconstrained solutions are assumed.

The rest of the paper is organized as follows. In \sref{sec:Model}, we present the system model and main assumptions used in the paper. In \sref{sec:Prob_Form}, we formulate the sparse channel estimation problem and present the idea of the proposed adaptive training/estimation algorithm. A hierarchical multi-resolution codebook for the  training precoders and combiner is then designed in \sref{sec:codebook}. Adaptive channel estimation algorithms are presented and discussed in \sref{sec:Algorithm}. The precoding design problem is formulated and a proposed hybrid RF/baseband precoding solution is presented in \sref{sec:Design}. In \sref{sec:Results}, simulation results demonstrating the performance of the proposed algorithms are given, before concluding the paper in  \sref{sec:conclusion}.

We use the following notation throughout this paper: $\bA$ is a matrix, $\ba$ is a vector, $a$ is a scalar, and $\cA$ is a set. $|\bA|$ is the determinant of $\bA$, $\|\bA \|_F$ is its Frobenius norm, whereas $\bA^T$, $\bA^H$, $\bA^*$, $\bA^{-1}$, $\pinv{\bA}$ are its transpose, Hermitian (conjugate transpose), conjugate, inverse, and pseudo-inverse respectively. $[\bA]_{\cR,:}$ $([\bA]_{:,\cR})$ are the rows (columns) of the matrix $\bA$ with indices in the set $\cR$, and $\mathrm{diag}(\ba)$ is a diagonal matrix with the entries of $\ba$ on its diagonal. $\bI$ is the identity matrix and $\mathbf{1}_{N}$ is the $N$-dimensional all-ones vector. $ \bA \circ \bB$ is the Khatri-Rao product of $\bA$, and $\bB$, $ \bA \otimes \bB$ is the Kronecker product of $\bA$, and $\bB$, and $\bA \odot \bB$ denotes the Hadamard product of $\bA$, and $\bB$. $\cN(\bm,\bR)$ is a complex Gaussian random vector with mean $\bm$ and covariance $\bR$. $\bbE\left[\cdot\right]$ is used to denote expectation.

\section{System Model} \label{sec:Model}

\begin{figure} [t]
\centerline{
\includegraphics[width=.65\columnwidth]{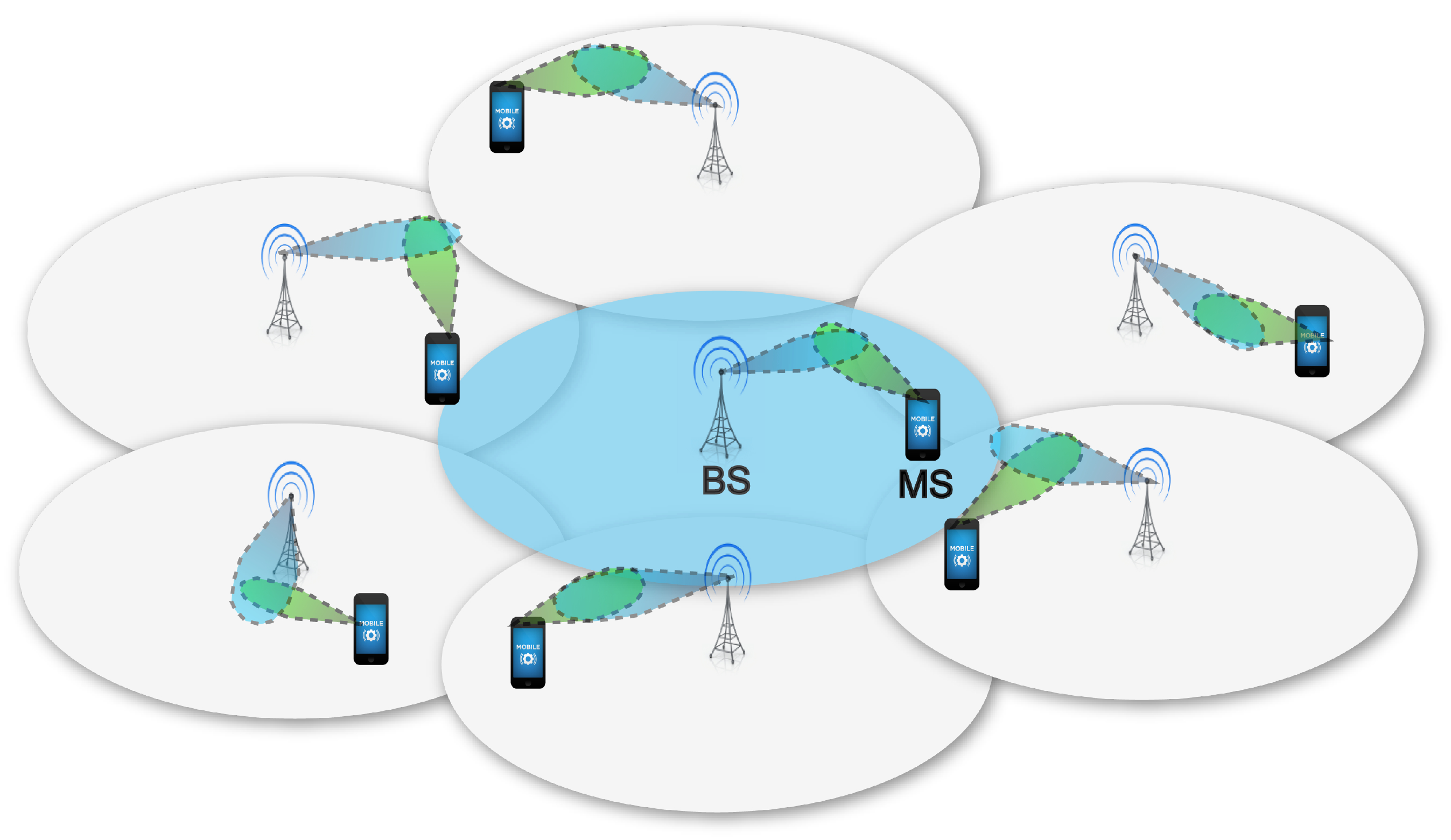}
}
\caption{A mmWave cellular system model, in which BSs and MSs communicate via directive beamforming using large antenna arrays}
\label{fig:ModelFig}
\end{figure}

Consider the mmWave cellular system shown in \figref{fig:ModelFig}. A BS with $N_\mathrm{BS}$ antennas and $\Nc$ RF chains is assumed to communicate with a single MS with $N_\mathrm{MS}$ antennas and $\Nc$ RF chains as shown in \figref{fig:BS_MU_arch_Fig}. The number of RF chains at the MSs is usually less than that of the BSs in practice, but we do not exploit this fact in our model. The BS and MS communicate via $N_\mathrm{S}$ data streams, such that $N_\mathrm{S} \leq \Nc \leq N_\mathrm{BS}$ and $N_\mathrm{S} \leq \Nc \leq N_\mathrm{MS}$~\cite{ayach2013spatially, samsung-practical-sdma-60GHz, xia2008multi}.

In this paper, we will focus on the downlink transmission. The BS is assumed to apply an $\Nc \times N_\mathrm{S}$ baseband precoder $\Fbb$ followed by an $N_\mathrm{BS} \times \Nc$ RF precoder, $\Frf$. If $\bF_\mathrm{T} =\Frf \Fbb$ is the $N_{BS} \times N_\mathrm{S}$ combined BS precoding matrix, the discrete-time transmitted signal is then

\begin{equation}
\bx=\bF_\mathrm{T} \bs,
\label{eq:signal_transmitted}
\end{equation}
where $\bs$ is the $N_\mathrm{S} \times 1$ vector of transmitted symbols, such that $\bbE\left[\bs\bs^H\right] = \frac{P_\mathrm{S}}{N_\mathrm{S}} \bI_{N_\mathrm{S}}$, and $P_\mathrm{S}$ is the average total transmit power. Since $\Frf$ is implemented using analog phase shifters, its entries are of constant modulus. We normalize these entries to satisfy  $\left|\left[\Frf\right]_{m,n}\right|^2=N_\mathrm{BS}^{-1}$, where $\left|\left[\Frf\right]_{m,n}\right|$ denotes the magnitude of the $(m,n)$th element of $\Frf$. The total power constraint is enforced by normalizing $\Fbb$ such that $\|\Frf \Fbb\|_F^2=N_\mathrm{S}$.

\begin{figure}
\centerline{
\includegraphics[width=1\columnwidth, height= .4\textwidth]{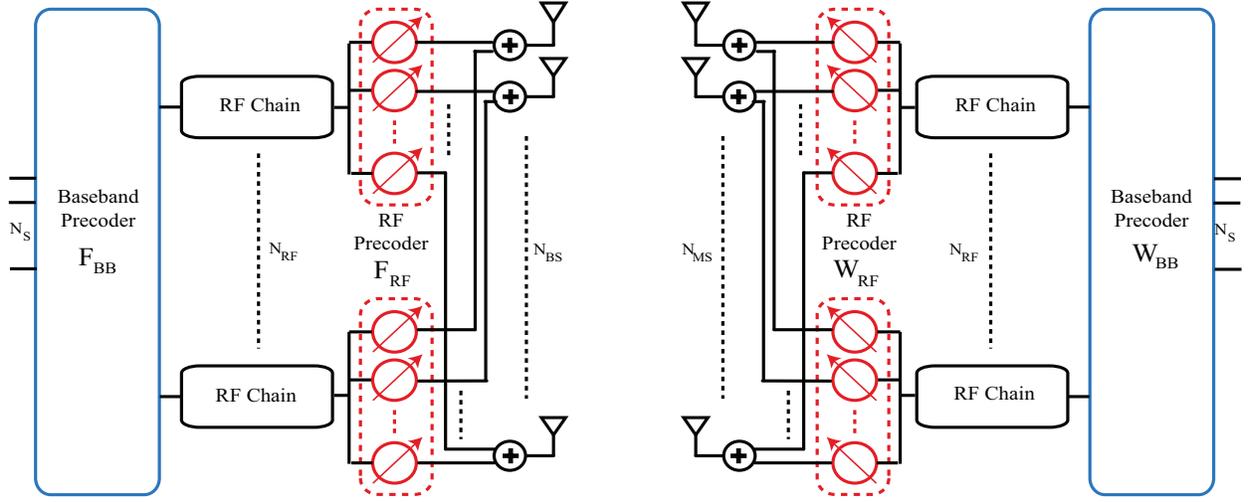}
}
\caption{Block diagram of BS-MS transceiver that uses RF and baseband beamformers at both ends.}
\label{fig:BS_MU_arch_Fig}
\end{figure}

We adopt a narrowband block-fading channel model in which an MS observes the received signal
\begin{equation}
\br= \bH \bF_\mathrm{T} \bs + \bn,
 \label{eq:received_signal}
\end{equation}
where $\bH$ is the $N_\mathrm{MS} \times N_\mathrm{BS}$ matrix that represents the mmWave channel between the BS and MS, and $\bn \sim \cN (0, \sigma^2 )$ is the Gaussian noise corrupting the received signal.

At the MS, the combiner $\bW_\mathrm{T}$ composed of the RF and baseband combiners $\Wrf$ and $\Wbb$ is used to process the received signal $\br$ which results in
\begin{equation}
\by= {\bW_\mathrm{T}}^H \bH\bF_\mathrm{T} \bs + {\bW_\mathrm{T}}^H \bn.
\label{eq:received_signal2}
\end{equation}

We will explain the proposed algorithms for the downlink model. The same algorithms, however, can be directly applied to the uplink system whose input-output relationship is identical to (\ref{eq:received_signal2}) with $\bH$ replaced by the uplink channel, and the roles of the precoders ($\Frf$, $\Fbb$) and combiners ($\Wrf$, $\Wbb$) switched.

While the mmWave channel estimation and precoding algorithms developed in the following sections consider only a BS-MS link with no interfering BSs, these algorithms will also be numerically evaluated by simulations in the case of mmWave cellular systems where out-of-cell interference exists in \sref{subsec:cellular_setup}.

Since mmWave channels are expected to have limited scattering~\cite{rappaport2012cellular, murdock201238, zhang2010channel,Ben-Dor, sayeed2007maximizing}, we adopt a geometric channel model with $L$ scatterers. Each scatterer is further assumed to contribute a single propagation path between the BS and MS~\cite{ayach2013spatially,raghavan2009multi}. Under this model, the channel $\bH$ can be expressed as
\begin{align}
\bH = \sqrt{\frac{N_\mathrm{BS} N_\mathrm{MS}}{\rho}} \sum_{\ell=1}^{L}\alpha_\ell \ba_\mathrm{MS}\left(\theta_\ell\right) \ba^H_\mathrm{BS}\left(\phi_\ell \right),
\label{eq:channel_model}
\end{align}
where $\rho$ denotes the average path-loss between the BS and MS, and $\alpha_\ell$ is the complex gain of the $\ell^\mathrm{th}$ path. The path amplitudes are assumed to be Rayleigh distributed, i.e., $\alpha_\ell \sim \cN\left(0, \bar{P}_\mathrm{R}\right), \ell=1,2,..., L$ with $\bar{P}_\mathrm{R}$ the average power gain. The variables $\phi_\ell \in [0, 2\pi]$ and $\theta_\ell \in [0, 2\pi]$ are the $\ell^\mathrm{th}$ path's azimuth angles of departure or arrival (AoDs/AoAs) of the BS and MS, respectively. Considering only the azimuth, and neglecting elevation, implies that all scattering happens in azimuth and that the BS and MS implement horizontal (2-D) beamforming only. Extensions to 3-D beamforming are possible~\cite{ayach2013spatially}. Finally, $\ba_\mathrm{BS}\left(\phi_\ell\right)$ and $\ba_\mathrm{MS}\left(\theta_\ell\right)$ are the antenna array response vectors at the BS and MS, respectively. While the algorithms and results developed in the paper can be applied to arbitrary antenna arrays, we use uniform linear arrays (ULAs), in the simulations of \sref{sec:Results}. If a ULA is assumed, $\ba_\mathrm{BS}\left(\phi_\ell\right)$ can be written as
\begin{equation}
\ba_\mathrm{BS}\left(\phi_\ell\right)= \frac{1}{\sqrt{N_\mathrm{BS}}} \left[ 1, e^{j {\frac{2 \pi}{\lambda}} d\sin\left(\phi_\ell\right)}, \hdots  ,e^{j\left(N_\mathrm{BS} -1\right){\frac{2\pi}{\lambda}}d\sin\left(\phi_\ell\right)} \right]^T,  \label{eq:received_signal3}
\end{equation}
where $\lambda$ is the signal wavelength, and $d$ is the distance between antenna elements. The array response vectors at the MS, $\ba_\mathrm{MS}\left(\theta_\ell\right)$, can be written in a similar fashion.

The channel in \eqref{eq:channel_model} is written in a more compact form as
\begin{equation}
\bH=\bA_\mathrm{MS} \mathrm{diag}\left(\boldsymbol\alpha\right) \bA_\mathrm{BS}^H, \label{eq:channel2}
\end{equation}
where $\boldsymbol\alpha=\sqrt{\frac{N_\mathrm{BS} N_\mathrm{MS}}{\rho}}  \left[\alpha_1, \alpha_2, ... ,\alpha_{L}\right]^T$. The matrices
\begin{align}
\bA_\mathrm{BS} &=\left[\ba_\mathrm{BS}\left({\phi}_1\right),\  \ba_\mathrm{BS}\left({\phi}_2\right),\ ... ,\ \ba_\mathrm{BS}\left({\phi}_{L}\right) \right],
\end{align}
and
\begin{align}
\bA_\mathrm{MS} & =\left[\ba_\mathrm{MS}\left({\theta}_1\right),\  \ba_\mathrm{MS}\left({\theta}_2\right),\  ...,\  \ba_\mathrm{MS}\left({\theta}_{L}\right) \right],
\end{align}
contain the BS and MS array response vectors.

In this paper, we assume that both the BS and MS have no a priori knowledge of the channel. Hence, in the first part of the paper, namely, \sref{sec:Prob_Form}-\sref{sec:Algorithm}, the mmWave channel estimation problem is formulated, and an adaptive CS based algorithm is developed and employed at the BS and MS to solve it. In the second part, i.e., \sref{sec:Design}, the estimated channel is used to construct the hybrid precoding and decoding matrices.

\section{Formulation of the MmWave Channel Estimation Problem} \label{sec:Prob_Form}
Given the geometric mmWave channel model in \eqref{eq:channel_model}, estimating the mmWave channel is equivalent to estimating the different parameters of the $L$ channel paths; namely the AoA, the AoD, and the gain of each path. To do that accurately and with low training overhead, the BS and MS need to carefully design their training precoders and combiners. In this section, we exploit the poor scattering nature of the mmWave channel, and formulate the mmWave channel estimation problem as a sparse problem. We will also briefly show how adaptive CS work invokes some ideas for the design of the training precoders and combiners. Inspired by these ideas, and using the hybrid analog/digital system architecture, we will develop a novel hierarchical multi-resolution codebook for the training beamforming vectors in \sref{sec:codebook}. We will then propose algorithms that adaptively use the developed codebook to estimate the mmWave channel along with evaluating their performance in \sref{sec:Algorithm}.


\subsection{A Sparse Formulation of the MmWave Channel Estimation Problem} \label{subsec:formulation}

Consider the system and mmWave channel models described in \sref{sec:Model}. If the BS uses a beamforming vector $\bff_p$, and the MS employs a measurement vector $\bw_q$ to combine the received signal, the resulting signal can be written as
\begin{equation}
y_{q,p}={\bw^H_q} \bH \bff_p s_p+ \bw_q^H \bn_{q,p},
\end{equation}
where $s_p$ is the transmitted symbol on the beamforming vector $\bff_p$, such that $\bbE\left[s_p s_p^H\right] = P$, with $P$ the average power used per transmission in the training phase. In \sref{sec:codebook}, we will develop a hybrid analog/digital design for the beamforming/measurement vectors, $\bff_p$ and $\bw_q$. If $M_\mathrm{MS}$ such measurements are performed by the MS vectors $\bw_q, q=1,2,...,M_\mathrm{MS}$ at $M_\mathrm{MS}$ successive instants to detect the signal transmitted over the beamforming vector $\bff_p$, the resulting vector will be
\begin{equation}
\by_p= \bW^H \bH \bff_p s_p+\mathrm{diag}{\left(\bW^H \left[\bn_{1,p},...,\bn_{M_\mathrm{MS},p}\right]\right)},
\end{equation}
where $\bW=\left[\bw_1, \bw_2, ..., \bw_{M_\mathrm{MS}}\right]$ is the  $N_\mathrm{MS}\times M_\mathrm{MS}$ measurement matrix. If the BS employs $M_\mathrm{BS}$ such beamforming vectors $\bff_p, p=1,..., M_\mathrm{BS}$, at $M_\mathrm{BS}$ successive time slots, and the MS uses the same measurement matrix $\bW$ to combine the received signal, the resultant matrix can then be written by concatenating the $M_\mathrm{BS}$ processed vectors $\by_p, p=1,2,...,M_\mathrm{BS}$
\begin{equation}
\bY=\bW^H \bH \bF \bS+\bQ,
\end{equation}
where $\bF=\left[\bff_1,\bff_2,..., \bff_{M_\mathrm{BS}}\right]$ is the $N_\mathrm{BS}\times M_\mathrm{BS}$ beamforming matrix used by the BS, and $\bQ$ is an $M_\mathrm{MS} \times M_\mathrm{BS}$ noise matrix given by concatenating the $M_\mathrm{BS}$ noise vectors. The matrix $\bS$ is a diagonal matrix carrying the $M_\mathrm{BS}$ transmitted symbols $s_p, p=1,...,M_\mathrm{BS}$ on its diagonal. For the training phase, we assume that all transmitted symbols are equal, namely, $\bS=\sqrt{P} \bI_{M_\mathrm{BS}}$ and therefore

\begin{equation}
\bY=\sqrt{P} \bW^H \bH \bF+\bQ.
\end{equation}

To exploit the sparse nature of the channel, we first vectorize the resultant matrix $\bY$
\begin{align}
\by_\mathrm{v} &=\sqrt{P}  \vec{\left(\bW^H \bH \bF\right)}+\vec{\left(\bQ\right)} \\
&\stackrel{(a)}{=}\sqrt{P} \left(\bF^T \otimes \bW^H \right)\vec(\bH)+\bn_\mathrm{Q} \\
&\stackrel{(b)}{=}\sqrt{P} \left(\bF^T \otimes \bW^H \right) \left({{\bA^*_\mathrm{BS}} \circ \bA_\mathrm{MS}}\right)\boldsymbol\alpha+\bn_\mathrm{Q}, \label{eq:vec_signal}
\end{align}
where $(a)$ follows from \cite[Theorem 13.26]{laub2004matrix}, $(b)$ follows from the channel model in \eqref{eq:channel2}, and the properties of the Khatri-Rao product, \cite{van2002optimum}. The matrix $\left( \bA_\mathrm{BS}^* \circ \bA_\mathrm{MS} \right)$ is an $N_\mathrm{BS} N_\mathrm{MS} \times L$ matrix in which each column has the form $\left( \ba_\mathrm{BS}^*\left(\phi_\ell\right) \otimes \ba_\mathrm{MS}\left(\theta_\ell\right) \right), \ell=1,2,..., L$, i.e., each column $\ell$ represents the Kronecker product of the BS and MS array response vectors associated with the AoA/AoD of the $\ell$th path of the channel.

To complete the problem formulation, we assume that the AoAs, and AoDs are taken from a uniform grid of $N$ points, with $N \gg L$, i.e., we assume that $\phi_\ell, \theta_\ell \in \left\{0,\frac{ 2 \pi}{N},..., \frac{2 \pi (N-1)}{N} \right\}, \ell=1,2,..., L$. As the values of the AoAs/AoDs are actually continuous, other off-grid based algorithms like sparse regularized total least squared \cite{sparse_Leus}, continuous basis pursuit \cite{ekanadham2011recovery}, or Newton refinement ideas \cite{ramasamy2013compressive} can be incorporated to reduce the quantization error. In this paper, we consider only the case of quantized AoAs/AoDs, leaving possible improvements for future work. We evaluate the impact of this quantization error on the performance of the proposed algorithms in this paper by numerical simulations in \sref{sec:Results}.

By neglecting the grid quantization error, we can approximate $\by_\mathrm{v}$ in \eqref{eq:vec_signal} as
\begin{align}
\by_\mathrm{v} = \sqrt{P} \left(\bF^T \otimes \bW^H \right) \bA_\mathrm{D} \bz+\bn_\mathrm{Q}, \label{eq:sparse_formulation}
\end{align}
where $\bA_\mathrm{D}$ is an $N_\mathrm{BS} N_\mathrm{MS} \times N^2$ dictionary matrix that consists of the $N^2$ column vectors of the form $\left( \ba_\mathrm{BS}^*\left(\bar{\phi}_u\right) \otimes \ba_\mathrm{MS}\left(\bar{\theta}_v\right)\right)$, with $\bar{\phi}_u$, and $\bar{\theta}_v$ the $u$th, and $v$th points, respectively, of the angles uniform grid, i.e, $\bar{\phi}_u = \frac{2 \pi u}{N}, u=0,2,...,N-1$, and $\bar{\theta}_v= \frac{2 \pi v}{N}, v=0,2,...,N-1$. $\bz$ is an $N^2 \times 1$ vector which carries the path gains of the corresponding quantized directions. Note that detecting the columns of $\bA_\mathrm{D}$ that correspond to non-zero elements of $\bz$, directly implies the detection of the AoAs and AoDs of the dominant paths of the channel. The path gains can be also determined by calculating the values of the corresponding elements in $\bz$.

The formulation of the vectorized received signal $\by_\mathrm{v}$ in \eqref{eq:sparse_formulation} represents a sparse formulation of the channel estimation problem as $\bz$ has only $L$ non-zero elements and $L \ll N^2$. This implies that the number of required measurements, $M_\mathrm{BS} M_\mathrm{MS}$, to detect the non-zero elements of $\bz$ is much less than $N^2$. In other words, this means that the BS does not need to transmit along each vector defined in the dictionary, nor does the MS need to observe signals using its entire codebook. Given this formulation in \eqref{eq:sparse_formulation}, CS tools can be leveraged to design estimation algorithms to determine the quantized AoAs/AoDs. If we define the sensing matrix $\boldsymbol\Psi$ as $\boldsymbol\Psi= \left(\bF^T \otimes \bW^H \right) \bA_\mathrm{D}$, the objective of the CS algorithms will be to efficiently design this sensing matrix to guarantee the recovery of the non-zero elements of the vector $\bz$ with high probability, and with a small number of measurements. One common criterion for that is the restricted isometry property (RIP), which requires the matrix $\boldsymbol\Psi^H \boldsymbol\Psi$ to be close to diagonal on average \cite{Rossi}.

To simplify the explanation of the BS-MS beamforming vectors' design problem in the later chapters, we prefer to use the Kronecker product properties and write \eqref{eq:sparse_formulation} as \cite{van2002optimum}
\begin{align}
\by_\mathrm{v}& = \sqrt{P} \left(\bF^T \bA_\mathrm{BS,D}^*  \otimes \bW^H \bA_\mathrm{MS,D} \right) \bz+\bn_\mathrm{Q} \\
&=\sqrt{P} \bF^T \bA_\mathrm{BS,D}^* \bz_\mathrm{BS} \otimes \bW^H \bA_\mathrm{MS,D} \bz_\mathrm{MS}+\bn_\mathrm{Q},  \label{eq:sparse_formulation3}
\end{align}
where $\bz_\mathrm{BS}$, and $\bz_\mathrm{MS}$ are two $N \times 1$ sparse vectors that have non-zero elements in the locations that correspond to the AoDs, and AoAs, respectively. $\bA_\mathrm{BS,D}$, and $\bA_\mathrm{MS,D}$ are $N_\mathrm{BS} \times N$, and $N_\mathrm{MS} \times N$ dictionary matrices that consist of column vectors of the forms $\ba_\mathrm{BS}\left(\bar{\phi}_u\right)$, and $\ba_\mathrm{MS}\left(\bar{\theta}_u\right)$, respectively.

In the standard CS theory, the number of measurement vectors required to guarantee the recovery of the $L$-sparse vector with high probability is of order $\cO(L\log(N/L))$ \cite{donoho2006compressed}. While these results are theoretically proved, their implementations to specific applications and the development of efficient algorithms require further work. We therefore resort to adaptive CS tools which invoke some ideas for the design of the training beamforming vectors.

\subsection{Adaptive Compressed Sensing Solution}\label{subsec:adaptive_sol}

In adaptive CS \cite{malloy2012near,malloy2012near2,Iwen}, the training process is divided into a number of stages. The training precoding, and measurement matrices used at each stage are not determined a priori, but rather depend on the output of the earlier stages. More specifically, if the training process is divided into $\mathrm{S}$ stages, then the vectorized received signals of these stages are
\begin{align}
\begin{split}
& \by_{(1)} =  \sqrt{P_{(1)}} \left(\bF_{(1)}^T \bA_\mathrm{BS,D}^*  \otimes \bW_{(1)}^H \bA_\mathrm{MS,D} \right) \bz+\bn_\mathrm{1} \\
& \by_{(2)} =  \sqrt{P_{(2)}} \left(\bF_{(2)}^T \bA_\mathrm{BS,D}^*  \otimes \bW_{(2)}^H \bA_\mathrm{MS,D} \right) \bz+\bn_\mathrm{2} \\
&\hspace{125pt} \vdots \\
&\by_{(\mathrm{S})} =  \sqrt{P_{(\mathrm{S})}} \left(\bF_{(\mathrm{S})}^T \bA_\mathrm{BS,D}^*  \otimes \bW_{(\mathrm{S})}^H \bA_\mathrm{MS,D} \right) \bz+\bn_\mathrm{S}
\end{split}\label{eq:sparse_formulation2}
\end{align}

The design of the $s$th stage training precoders and combiners, $\bF_{(\mathrm{s})}, \bW_{(\mathrm{s})}$, depends on $\by_{(1)}, \by_{(2)}$, $ ...$, $\by_{(s-1)}$. Recent research in \cite{malloy2012near,malloy2012near2,Iwen} shows that adaptive CS algorithms yield better performance than standard CS tools at low SNR, which is the typical case at mmWave systems before beamforming. Moreover, these adaptive CS ideas that rely on successive bisections provide important insights that can be used in the design of the training beamforming vectors.

In our proposed channel estimation algorithm described in \sref{sec:Algorithm}, the training beamforming vectors are adaptively designed based on the bisection concept. In particular, the algorithm starts initially by dividing the vector $\bz$ in \eqref{eq:sparse_formulation2} into a number of partitions, which equivalently divides the AoAs/AoDs range into a number of intervals, and design the training precoding and combining matrices of the first stage, $\bF_{(\mathrm{1})}, \bW_{(\mathrm{1})}$, to sense those partitions. The received signal $\by_{(1)}$ is then used to determine the partition(s) that are highly likely to have non-zero element(s) which are further divided into smaller partitions in the later stages until detecting the non-zero elements, the AoAs/AoDs, with the required resolution. If the number of BS precoding vectors used in each stage of the adaptive algorithm equals $K$, where $K$ is a design parameter, then the number of adaptive stages needed to detect the AoAs/AoDs with a resolution $\frac{2 \pi}{N}$ is $\mathrm{S}=\log_K{N}$, which we assume to be integer for ease of exposition. Before delving into the details of the algorithm, we will focus in the following section on the design of a multi-resolution beamforming codebook which is essential for the proper operation of the adaptive channel estimation algorithm.

\section{Hybrid Precoding Based Multi-Resolution Hierarchical Codebook} \label{sec:codebook}
In this section, we present a novel hybrid analog/digital based approach for the design of a multi-resolution beamforming codebook. Besides considering the RF limitations, namely, the constant amplitude phase shifters with quantized phases, the proposed approach for constructing the beamforming vectors is general for ULAs/non-ULAs, has  a very-low complexity, and outperforms the analog-only beamforming codebooks thanks to its additional digital processing layer.

The design of a multi-resolution or variant beamwidth beamforming vector codebook has been studied before in \cite{tsang2011successive,Multilevel,chen2011multi,Wang1,van2002optimum}. This prior work focused on analog-only beamforming vectors, and on the physical design of the beam patterns. Unfortunately, the design of analog-only multi-resolution codebooks is subject to practical limitations in mmWave. (1) The existence of quantized phase shifters makes the design of non-overlapping beam patterns difficult, and may require an exhaustive search over a large space given the large number of antennas. (2) The design of analog-only beamforming vectors with certain beamwidths relies mostly on the beamsteering beam patterns of ULAs, and is hard to apply for non-ULAs due to the lack of intuition about their beam patterns.

To simplify explaining the codebook structure and design, we focus on the design of the BS training precoding codebook $\cF$; a similar approach can be followed to construct the MS training codebook $\cW$.

\subsection{Codebook Structure}\label{subsec:Structure}
The proposed hierarchical codebook consists of $\mathrm{S}$ levels, $\cF_s, s=1,2,..., \mathrm{S}$. Each level contains beamforming vectors with  a certain beamwidth to be used in the corresponding training stage of the adaptive mmWave channel estimation algorithm. \figref{fig:HP_codebook} shows the first three levels of an example codebook with $N=256$, and $K=2$, and \figref{fig:HP_codebook2} illustrates the beam patterns of the beamforming vectors of each codebook level.

\begin{figure}
\centerline{
\includegraphics[width=5 in, height= .3\textwidth]{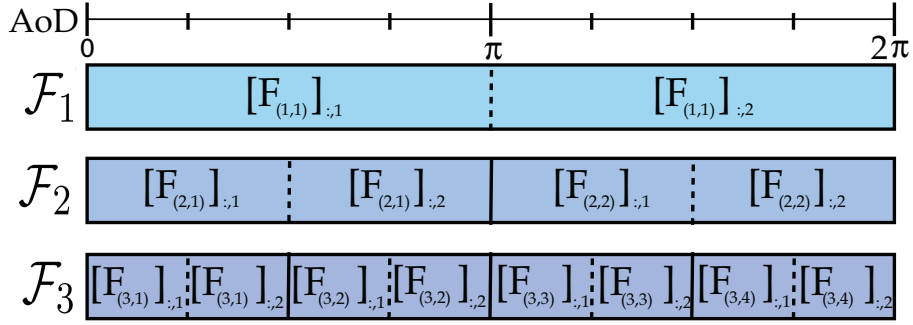}
}
\caption{An example of the structure of a multi-resolution codebook with a resolution parameter $N=8$, and with $K=2$ beamforming vectors in each subset.}
\label{fig:HP_codebook}
\end{figure}

\begin{figure}
\centerline{
\includegraphics[ height= .3\textwidth]{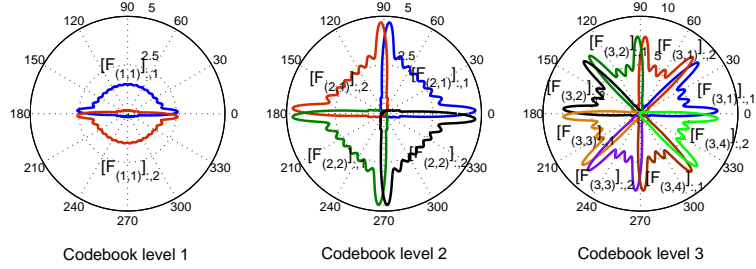}
}
\caption{The resulting beam patterns of the beamforming vectors in each codebook level.}
\label{fig:HP_codebook2}
\end{figure}

In each codebook level $s$, the beamforming vectors are divided into $K^{s-1}$ subsets, with $K$ beamforming vectors in each of them. Each subset $k$, of the codebook level $s$ is associated with a unique range of the AoDs equal to $\{\frac{2 \pi u}{N}\}_{u \in \cI_{\left(s,k\right)}}$, where $\cI_{\left(s,k\right)}=\left\{\frac{(k-1)N}{K^{s-1}},...,\frac{kN}{K^{s-1}}\right\}$. This AoD range is further divided into $K$ sub-ranges, and each of the $K$ beamforming vectors in this subset is designed so as to have an almost equal projection on the vectors $\ba_\mathrm{BS}\left(\bar{\phi}_u\right)$, with $u$ in this sub-range, and zero projection on the other vectors. Physically, this implies the implementation of a beamforming vector with a certain beamwidth determined by these sub-ranges, and steered in pre-defined directions.
\subsection{Design of the Codebook Beamforming Vectors}\label{subsec:codebook_design}

In each codebook level $s$, and subset $k$, the beamforming vectors $\left[\bF_{(s,k)}\right]_{:,m}, m=1,2,...,K$ are designed such that
\begin{equation}
\left[\bF_{(s,k)}\right]_{:,m}^H \ba_\mathrm{BS}\left(\bar{\phi}_u\right) = \left\{\begin{array}{ll} C_s & \mbox{if } u \in \cI_{(s,k,m)} \\ 0 & \mbox{if } u \not\in \cI_{(s,k,m)} \end{array}\right.,
\label{eq:Des1}
\end{equation}
where $\cI_{(k,s,m)}=\left\{\frac{N}{K^{s}}\left(K(k-1)+m-1\right)+1,..., \frac{N}{K^{s}}\left(K(k-1)+m\right)\right\}$ defines the sub-range of AoDs associated with the beamforming vector $\left[\bF_{(s,k)}\right]_{:,m}$, and $C_s$ is a normalization constant that satisfies $\|\bF_{(s,k)}\|_F=K$. For example, the beamforming vector $\left[\bF_{(2,1)}\right]_{:,1}$ in \figref{fig:HP_codebook} is designed such that it has a constant projection on the array response vectors $\ba_\mathrm{BS}\left(\bar{\phi}_u\right)$, $u$ is in $\left\{0, 1, ..., 63\right\}$, i.e., $\bar{\phi}_u$ is in $\left\{0, ..., 2 \pi \frac{63}{256}\right\}$, and zero projection on the other directions.

In a more compact form, we can write the design objective of the beamforming vectors $\bF_{(s,k)}$ in \eqref{eq:Des1} as the solution of
\begin{equation}
\bA_\mathrm{BS,D}^H \bF_{(s,k)}=C_s \bG_{(s,k)},
\label{eq:Des2}
\end{equation}
where $\bG_{(s,k)}$ is an $N \times K$ matrix where each column $m$ containing $1's$ in the locations $u, u \in \cI_{(s,k,m)}$, and zeros in the locations $u, u \not\in \cI_{(s,k,m)}$.
Now, we note that the BS AoDs matrix $\bA_\mathrm{BS,D}$ is an over-complete dictionary with $N \geq N_\mathrm{BS}$, i.e., \eqref{eq:Des2} represents an inconsistent system of which the approximate solution is given by $\bF_{(s,k)}=C_s (\bA_\mathrm{BS,D} \bA_\mathrm{BS,D}^H)^{-1} \bA_\mathrm{BS,D} \bG_{(s,k)}$. Further, given the available system model in \sref{sec:Model},
the precoding matrix $\bF_{(s,k)}$ is defined as $\bF_{(s,k)}=\bF_{\mathrm{RF},(s,k)} \bF_{\mathrm{BS},(s,k)}$. As each beamforming vector will be individually used in a certain time instant, we will design each of them independently in terms of the hybrid analog/digitl precoders. Consequently, the design of the hybrid analog and digital training precoding matrices is accomplished by solving

\begin{align}
\begin{split}
\left\{\bF_{\mathrm{RF},(s,k)}^{\star}, \left[\bF_{\mathrm{BB},(s,k)}^{\star}\right]_{:,m}\right\} &  =  \arg\min \ \  \|\left[\bF_{(s,k)}\right]_{:,m} -\bF_{\mathrm{RF},(s,k)}\left[\bF_{\mathrm{BB},(s,k)}\right]_{:,m}\|_F, \\
& \hspace{15pt} \mathrm{s.t}. \ \  \left[\bF_{\mathrm{RF},(s,k)}\right]_{:,i} \in \left\{\left[\bA_\mathrm{can}\right]_{:,\ell} | \ 1 \leq \ell \leq N_\mathrm{can}\right\}, i=1,2,..., N_\mathrm{RF} \\
& \hspace{39pt} \|\bF_{\mathrm{RF},(s,k)}\left[\bF_{\mathrm{BB},(s,k)}\right]_{:,m}\|_F^2=1,
\label{eq:BF_Design}
\end{split}
\end{align}
where $\left[\bF_{(s,k)}\right]_{:,m}=C_s (\bA_\mathrm{BS,D} \bA_\mathrm{BS,D}^H)^{-1} \bA_\mathrm{BS,D} \left[ \bG_{(s,k)}\right]_{:,m}$, and $\bA_\mathrm{can}$ is an $N_\mathrm{BS} \times N_\mathrm{can}$ matrix which carries the finite set of possible analog beamforming vectors. The columns of the candidate matrix
$\bA_\mathrm{can}$ can be chosen to satisfy arbitrary analog beamforming constraints. Two example candidate beamformer designs we consider in the simulations of Section \ref{sec:Results}
are summarized as follows.
\begin{enumerate}
\item Equally spaced ULA beam steering vectors~\cite{ayach2013spatially}, i.e., a set of $N_\mathrm{can}$ vectors of the form $\ba_{BS}(\frac{t_\mathrm{can} \pi}{N})$ for $t_\mathrm{can}=0,\ 1,\ 2,\ \hdots, N_\mathrm{can}-1$.
\item {Beamforming vectors whose elements can be represented as quantized phase shifts. In the case of quantized phase shifts, if each phase shifter is controlled by an $N_\mathrm{Q}$-bit input, the entries of the
candidate precoding matrix $\bA_\mathrm{can}$ can all be written as $e^{j\frac{k_\mathrm{Q} 2\pi}{2^{N_\mathrm{Q}}}}$ for some $k_\mathrm{Q}=0,\ 1,\ 2,\ \hdots, 2^{N_\mathrm{Q}}-1$.}
\end{enumerate}

Now, given the matrix of possible analog beamforming vectors $\bA_\mathrm{can}$, the optimization problem in \eqref{eq:BF_Design} can be reformulated as a sparse approximation problem \cite{ayach2013spatially,alkhateeb}
\begin{align}
\begin{split}
\left[\bF_{\mathrm{BB},(s,k)}^{\star}\right]_{:,m} &  =  \arg\min \ \  \|\left[\bF_{(s,k)}\right]_{:,m}-\bA_\mathrm{can} \left[\bF_{\mathrm{BB},(s,k)}\right]_{:,m}\|_F, \\
&\hspace{15pt}  \mathrm{s.t}. \ \ \|\mathrm{diag}\left({\left[\bF_{\mathrm{BB},(s,k)}\right]_{:,m} \left[\bF_{\mathrm{BB},(s,k)}\right]^H_{:,m}}\right)\|_{\ell_0}=\Nc.\\
& \hspace{39pt} \|\bF_{\mathrm{RF},(s,k)}\left[\bF_{\mathrm{BB},(s,k)}\right]_{:,m}\|_F^2=1.
\label{eq:BF_Design2}
\end{split}
\end{align}

Note that the first constraint in \eqref{eq:BF_Design2} ensures that only $\Nc$ rows of $\left[\bF_{\mathrm{BB},(s,k)}\right]_{:,m}$ can be non-zeros. Hence, after the design of the baseband training precoder
using this sparse problem, the columns of $\bA_\mathrm{can}$ that correspond to the non-zero rows of $\left[\bF_{\mathrm{BB},(s,k)}\right]_{:,m}$ are chosen to be the RF precoder $\bF_{\mathrm{RF},(s,k)}$. 

The exact solution of the sparse approximation problem in \eqref{eq:BF_Design2} requires solving a combinatorial optimization problem of high complexity. Hence, following \cite{ayach2013spatially},
we develop an orthogonal matching pursuit algorithm to iteratively solve this problem as shown in Algorithm \ref{alg3}. Also, note that the constant $C_s$ is not known a priori, and
should be ideally maximized as it is proportional to the beamforming gain as indicated in \eqref{eq:Des1}. However, and for the sake of a low-complexity solution, we will assume that it
is a constant, and calculate its value after the design of the beamforming vectors to normalize them as shown in Algorithm \ref{alg3}.

\begin{figure}
\centerline{
\includegraphics[ height= .3\textwidth]{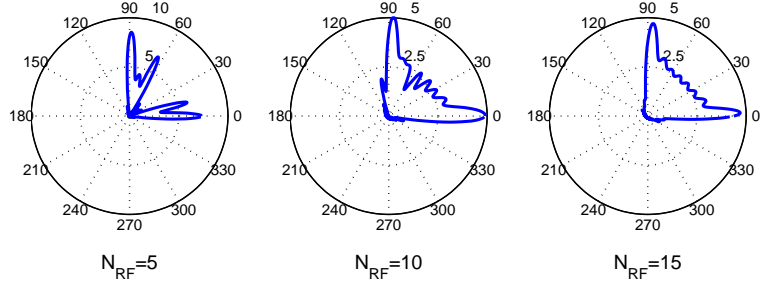}
}
\caption{Beam patterns approximation with different numbers of RF chains.}
\label{fig:RF_patterns}
\end{figure}

\begin{algorithm} [!t]                     
\caption{Hybrid Analog-Digital Training Precoders Design}          
\label{alg3}                           
\begin{algorithmic}                    
    \State $\cR = \phi$
    \State $\bff_\mathrm{res}$=$(\bA_\mathrm{BS,D} \bA_\mathrm{BS,D}^H)^{-1} \bA_\mathrm{BS,D}  \left[\bG_{(s,k)}\right]_{:,m}$
    \State $\bff^\star=(\bA_\mathrm{BS,D} \bA_\mathrm{BS,D}^H)^{-1} \bA_\mathrm{BS,D} \left[ \bG_{(s,k)}\right]_{:,m}$
    \For{$i \leq \Nc$}
        \State $\boldsymbol\Phi = \bff_\mathrm{res}^H \bA_\mathrm{can} $
        \State $n=\arg\max_{n=1,2,..N_\mathrm{can}} \left[\boldsymbol\Phi^H \boldsymbol\Phi\right]_{i,i}$
        \State $\cR = \cR \cup n$
        \State $\bF_{\mathrm{RF},(s,k)} = \left[\bA_\mathrm{can}\right]_{:,\cR}$
        \State $\left[\bF_{\mathrm{BB},(s,k)}\right]_{:,m}={\left({\bF_{\mathrm{RF},(s,k)}}^H \bF_{\mathrm{RF},(s,k)}\right)}^{-1} \bF_{\mathrm{RF},(s,k)}^H \bff^\star$
        \State $\bff_\mathrm{res}=\frac{\bff_\mathrm{res}-\bF_{\mathrm{RF},(s,k)} \left[\bF_{\mathrm{BB},(s,k)}\right]_{:,m}}{\|\bff_\mathrm{res}-\bF_{\mathrm{RF},(s,k)} \left[\bF_{\mathrm{BB},(s,k)}\right]_{:,m}\|_F}$
    \EndFor
        \State $C_s=\sqrt{\frac{1}{\|\bF_{\mathrm{RF},(s,k)} \left[\bF_{\mathrm{BB},(s,k)}\right]_{:,m}\|_F}}$
        \State $\left[\bF_{\mathrm{BB},(s,k)}\right]_{:,m}=C_s \left[\bF_{\mathrm{BB},(s,k)}\right]_{:,m}$
    \end{algorithmic}
\end{algorithm}

In summary, Algorithm \ref{alg3} starts by finding the vector $\left[\bA_\mathrm{can}\right]_{:,l}$ along which the matrix $\bF^\star$ has the maximum projection.
It then appends the selected column vector $\left[\bA_\mathrm{can}\right]_{:,n}$ to the RF precoder $\bF_{\mathrm{RF},(s,k)}$. After the dominant vector is found and the least squares solution to $\left[\bF_{\mathrm{BB},(s,k)}\right]_{:,m}$ is
calculated, the contribution of the selected vector is removed and the algorithm proceeds to find the column along which the ``residual precoding matrix'' $\bF_\mathrm{res}$ has the largest projection. The process continues
until all $\Nc$ beamforming vectors have been selected. At the end of the $\Nc$ iterations, the algorithm would have: (i) constructed an $N_\mathrm{BS} \times \Nc$ RF training beamforming matrix $\bF_{\mathrm{RF},(s,k)}$,
(ii) found the baseband training precoder $\left[\bF_{\mathrm{BB},(s,k)}\right]_{:,m}$ which minimizes the objective in (\ref{eqn:combinationprecoding}), and (iii) calculated the value of the constant $C_s$. It is worth mentioning here that if the K BS training vectors defined by $\bF_{(s,k)}$ will not be jointly used as suggested by the formulation in \eqref{eq:sparse_formulation}, then Algorithm \ref{alg3} should be used to separately design each vector alone to approximate the corresponding vector of $\bF^\star$. This yields a better approximation as all the RF chains will be explicitly used for this vector in the approximation problem.

An example of the beam patterns resulting from applying the proposed algorithm is shown in \figref{fig:RF_patterns}. These patterns are generated by a BS has $32$ antennas, and a number of RF chains $N_\mathrm{RF}=5,10,15$ to approximate the beamforming vectors $\left[\bF_{(2,1)}\right]_{:,1}$ shown in \figref{fig:HP_codebook2}.

After the design of the BS training beamforming vectors for the $k$th subset of the $s$th codebook, the following quantities are calculated, as they will be used after that in the channel estimation algorithm in \sref{sec:Algorithm}:
\begin{itemize}
\item{\textbf{Beamforming Gain}: Given the channel model in \eqref{eq:channel_model}, and the codebook beamforming design criteria in \eqref{eq:Des1}, we define the beamforming gain of the BS training vectors at the $s$th stage as
$G_{(s)}^\mathrm{BS}=N_\mathrm{BS} C_s^2$. A similar definition can be used for the MS beamforming vectors, yielding a total training beamforming gain at the $s$th stage equal to $G_{(s)}=G_{(s)}^\mathrm{BS} G_{(s)}^\mathrm{MS}$.
}
\item{\textbf{Error Matrix}: As the system in \eqref{eq:Des2} is inconsistent, the solution given by the pseudo-inverse means that $\bA_\mathrm{BS,D}^H \bF_{(s,k)}$ may not be exactly equal to $C_s \bG_{(s,k)}$.
Moreover, the limitations of the RF beamforming vectors, and the approximate solution of the sparse approximation problem in \eqref{eq:BF_Design2} results in an additional error in satisfying \eqref{eq:Des2}. This error physically means (i) the existence of a spectral leakage of the beamforming vectors outside their supposed AoD sub-ranges, and (ii) the beamforming gain is not exactly uniform over the desired AoD ranges. To take the effect of this error into the performance analysis of the proposed channel estimation algorithm in \sref{sec:Algorithm}, we define the error matrix of each subset $k$ of the $s$th BS beamforming codebook level as
\begin{equation}
\bE_{(s,k)}^\mathrm{BS}=\bA_\mathrm{BS,D}^H \bF_{(s,k)} - C_s \bG_{(s,k)}.
\end{equation}

As a similar error exists in the MS combining codebook, we can define the final error experienced by the received vector $\by_s$ in \eqref{eq:sparse_formulation3} after applying the Kronecker product as
\begin{equation}
\bE_{(s,k_\mathrm{BS},k_\mathrm{MS})}={\bE_{(s,k_\mathrm{BS})}^\mathrm{BS}}^T \otimes {\bE_{(s,k_\mathrm{MS})}^\mathrm{MS}}^T+{\bE_{(s,k_\mathrm{BS})}^\mathrm{BS}}^T \otimes C_s^\mathrm{MS}  \bG_{(s,k_\mathrm{MS})}^T+ C_s^\mathrm{BS} \bG_{(s,k_\mathrm{BS})}^T \otimes {\bE_{(s,k_\mathrm{MS})}^\mathrm{MS}}^T.
\end{equation}

Now, if we also define a new matrix $\bG_{(s,k_\mathrm{BS},k_\mathrm{MS})}=\bG_{(s,k_\mathrm{BS})} \otimes \bG_{(s,k_\mathrm{MS})}$, then we can rewrite
 $\by_s$ in \eqref{eq:sparse_formulation2} as follows assuming the subsets $k_\mathrm{BS}$ of $\cF_s$, and $k_\mathrm{MS}$ of $\cW_s$ are used at the BS and MS
\begin{align}
\by_{(s)} &= \sqrt{P_{(s)}} \left(\bF_{(s)}^T \otimes \bW_s^H \right) \bA_\mathrm{D} \bz+\bn_\mathrm{Q} \\
&= \sqrt{P_{(s)}} \left( \bF_{(s)}^T \bA_\mathrm{BS,D} \otimes \bW_{(s)}^H \bA_\mathrm{MS,D} \right)\bz+\bn_\mathrm{Q} \\
&= \sqrt{P_{(s)}} \left( \sqrt{\frac{G_{(s)}}{N_\mathrm{BS} N_\mathrm{MS}}}\bG_{(s,k_\mathrm{BS},k_\mathrm{MS})}+\bE_{(s,k_\mathrm{BS},k_\mathrm{MS})} \right)\bz+\bn_\mathrm{Q}.\label{eq:sparse4}
\end{align}
}

\item{\textbf{Forward and Backward Gains}: To include the effect of the previously defined error matrix on the beamforming gain, we will define the forward and backward gains
of the designed beamforming vectors. First, note that $\bG_{(s,k_\mathrm{BS},k_\mathrm{MS})}$ is a $K^2 \times N^2$ matrix in which each row corresponds to a certain pair of the
precoding/measurement vectors, and each column corresponds to a certain quantized AoA/AoD pair. We denote each AoA/AoD pair as a direction $d$, $d=1,2,...,N^2$. Now, we note that if a certain direction $d$ lies in the AoAs/AoDs range defined by a certain precoder/measurement vector $m, m=1,2,...,K^2$, then the entry $\left[\bG_{(s,k_\mathrm{BS},k_\mathrm{MS})}\right]_{m,d}=1$ to indicate that this direction lies in the main lobe of the patterns of both the precoder and measurement vectors. We also notice that each column $d$ of the matrix $\bG_{(s,k_\mathrm{BS},k_\mathrm{MS})}$ contains only one non-zero value, equal to 1, as each direction can not lie in the main lobe, i.e., AoAs/AoDs range, defined by more than one precoding/measurement vectors; thanks to the non-overlapping design of the beamforming vectors in \eqref{eq:Des1}.

We can then define the forward gain in the direction $d, d=1,2,..., N^2$, when the subsets $k_\mathrm{BS}, k_\mathrm{MS}$ of the codebooks $\cF_s, \cW_s$ are used as
\begin{equation}
G^\mathrm{F}_{(s,k_\mathrm{BS},k_\mathrm{MS},m,d)}=\left|\sqrt{G_{(s)}}+\sqrt{N_\mathrm{BS}N_\mathrm{MS}} \left[\bE_{s,k_\mathrm{BS},k_\mathrm{MS}}\right]_{m(d),d}\right|^2,
\end{equation}
where $m(d)$ is defined as $m(d) \in \left\{{m=1,2,..,K^2|\left[\bG_{(s,k_\mathrm{BS},k_\mathrm{MS})}\right]_{m,d}=1}\right\}$ which corresponds to only one value, i.e., one precoding/measurement pair, as described earlier.

To define the backward gain in a certain direction $d$, we need to specify also the beamforming/measurement pair $\bar{m}\in \cG_{\bar{m}}=\left\{m=1,2,..,K^2|\left[\bG_{(s,k_\mathrm{BS},k_\mathrm{MS})}\right]_{m,d}=0\right\}$.
Hence, we define the backward gain as
\begin{equation}
G^\mathrm{B}_{(s,k_\mathrm{BS},k_\mathrm{MS},d,\bar{m})}=N_\mathrm{BS}N_\mathrm{MS} \left|\left[\bE_{s,k_\mathrm{BS},k_\mathrm{MS}}\right]_{\bar{m},d}\right|^2.
\end{equation}

Finally, we define the ratio between the forward and backward beamforming gains in a certain direction, $d$, due to precoding/measurment pairs $m(d), \bar{m}$ as
\begin{equation}
\beta_{(s,k_\mathrm{BS},k_\mathrm{MS},d,\bar{m})}=\frac{G^\mathrm{F}_{(s,k_\mathrm{BS},k_\mathrm{MS},d)}}{G^\mathrm{B}_{(s,k_\mathrm{BS},k_\mathrm{MS},d,\bar{m})}}.
\end{equation}
}
\end{itemize}

Note that one disadvantage of the proposed approach for constructing the beamforming vectors is the shown ripples in the main lobes in \figref{fig:HP_codebook2}. This comes from the approximate solution of the inconsistent system in \eqref{eq:Des2}, and from the fact that we design over a finite set of directions in $\bA_\mathrm{BS,D}$. These patterns, however, are acceptable for the \textit{sparse} channel estimation problem that we consider. The main reason is that this ripple is in the main lobe, while the side lobes in these patterns are very small. If the channel has only one path in a certain direction $d$, then it will be affected by only one sample of this main lobe, which is in the direction $d$. Hence, this ripple in the main lobe just affects the forward beamforming gain. As we will show in the analysis of the proposed adaptive channel estimation algorithm in \sref{sec:Algorithm}, the performance of the proposed algorithm depends mainly on the ratio of the backward to forward gains. Therefore, the small side lobes, i.e., backward gains, greatly reduce the impact of these fluctuations in the main lobe on the overall performance of the proposed sparse channel estimation algorithms.

\section{Adaptive Estimation Algorithms for MmWave Channels}\label{sec:Algorithm}

In this section, we consider the sparse channel estimation problem formulated in \eqref{eq:sparse_formulation2} of \sref{sec:Prob_Form}, and propose algorithms that adaptively use the hierarchical codebook
developed in \sref{sec:codebook} to estimate the mmWave channel. We firstly address this problem for the rank-one channel model, i.e., when the channel has only one-path, in \sref{subsec:single-path}. We then extend the proposed algorithm for the multi-path case in \sref{subsec:multi-path}.
\subsection{Adaptive Channel Estimation Algorithm for Single-Path MmWave Channels}\label{subsec:single-path}
Given the problem formulation in \eqref{eq:sparse_formulation2}, the single-path channel implies that the vector $\bz$ has only one non-zero element. Hence, estimating the single-path channel
is accomplished by determining the location of this non-zero element, which in turn defines the AoA/AoD, and the value of this element, which decides the channel path gain.
To efficiently do that with low training overhead, we propose Algorithm \ref{alg1} which adaptively searches for the non-zero element of $\bz$ by using the multi-resolution beamforming vectors designed in \sref{sec:codebook}.

\begin{algorithm} [!t]                     
\caption{Adaptive Estimation Algorithm for Single-Path MmWave Channels}          
\label{alg1}                           
\begin{algorithmic}                    
    \State \textbf{Input:} BS and MS know $N,K$, and have $\cF, \cW$.
    \State \textbf{Initialization:} $k_1^\mathrm{BS}=1, k_1^\mathrm{MS}=1$  // Initialize the subsets to be used of codebooks $\cF, \cW$
    \State $\mathrm{S}=\log_K{N}$ // The number of adaptive stages
    \For{$s \leq \mathrm{S}$}
        \For {$m_\mathrm{BS} \leq K$}
            \State BS transmits a training symbol using $\left[\bF_{(s,k_s^\mathrm{BS})}\right]_{:,m_\mathrm{BS}}$
            \For {$m_\mathrm{MS} \leq K$}
                \State MS makes a measurement using $\left[\bW_{(s,k_s^\mathrm{MS})}\right]_{:,m_\mathrm{MS}}$
            \EndFor
            \State After MS measurements: $\by_{m_\mathrm{BS}}= \sqrt{P_{s}}\left[\bW_{(s,k_s^\mathrm{MS})}\right] \bH \left[\bF_{(s,k_s^\mathrm{BS})}\right]_{:,m_\mathrm{BS}}+\bn_{m_\mathrm{BS}}$
        \EndFor
        \State $\bY_{(s)}=[\by_1, \by_2, ..., \by_K]$
        \State $\left(m_\mathrm{BS}^{\star},m_\mathrm{MS}^{\star}\right)=\arg\max_{ \forall m_\mathrm{BS}, m_\mathrm{MS}=1,2,..., K} \left[\bY_{(s)} \odot \bY_{(s)}^* \right]_{m_\mathrm{MS},m_\mathrm{BS}}$
        \State $k_{s+1}^\mathrm{BS}=K(m_\mathrm{BS}^\star-1)+1, k_{s+1}^\mathrm{MS}=K(m_\mathrm{MS}^\star-1)+1$
    \EndFor
    \State $\hat{\phi}=\bar{\phi}_{{k_{\mathrm{S}+1}^\mathrm{BS}}}, \hat{\theta}=\bar{\theta}_{k_{\mathrm{S}+1}^{\mathrm{MS}}}$
    \State $\hat{\alpha}=\sqrt{\frac{\rho}{P_{(\mathrm{S})}G_\mathrm{(S)} }}{\left[\bY_{(\mathrm{S})}\right]_{m_\mathrm{MS}^{\star},m_\mathrm{BS}^{\star}}}$
    \end{algorithmic}
\end{algorithm}

Algorithm \ref{alg1} operates as follows. In the initial stage, the BS uses the $K$ training precoding vectors of the first level of the codebook $\cF$ in \sref{sec:codebook}. For each of those vectors, the MS uses the $K$ measurement vectors of the first level of $\cW$ to combine the received signal. Note that the first level of the hierarchical codebook in \sref{sec:codebook} has only one subset of beamforming vectors. After the $K^2$ precoding-measurement steps of this stage, the MS compares the power of the $K^2$ received signals to determine the one with the maximum received power. As each one of the precoding/measurement vectors is associated with a certain range of the quantized AoA/AoD, the operation of the first stage divides the vector $\bz$ in \eqref{eq:sparse_formulation2} into $K^2$ partitions, and compares between the power of the sum of each of them. Hence, the selection of the maximum power received signal implies the selection of the partition of $\bz$, and consequently the range of the quantized AoA/AoD, that is highly likely to contain the single path of the channel. The output of the maximum power problem is then used to determine the subsets of the beamforming vectors of level $s+1$ of $\cF$, and $\cW$ to be used in the next stage. The MS then feeds back the selected subset of the BS precoders to the BS to use it in the next stage, which needs only $\log_2{K}$ bits. As the beamforming vectors of the next levels have higher and higher resolution, the AoA/AoD ranges are further refined adaptively as we proceed in the algorithm stages until the desired resolution, $\frac{2 \pi}{N}$, is achieved. Note that the training powers in the $\mathrm{S}$ stages are generally different as will be discussed shortly.

Based on the proposed algorithm, the total number of stages required to estimate the AoA/AoD with a resolution $\frac{2 \pi}{N}$ is $\log_K{N}$. Also, since we need $K$ beamforming vectors, and $K$ measurement vectors for each beamforming vector in each stage, the total number of steps needed to estimate the mmWave channel using the proposed algorithm becomes $K^2 \log_K{N}$ steps. Moreover, since $\Nc$ RF chains can be simultaneously used at the MS to combine the measurements, the required number of steps can be further reduced to be $ K\lceil\frac{K}{\Nc}\rceil \log_K{N}$.

In the following theorem, we characterize the performance of the proposed algorithm for the case of single dominant path channels, i.e., assuming that the channel model in \eqref{eq:channel_model} has $L=1$. We find an upper bound of the probability of error in estimating the AoA/AoD with a certain resolution using Algorithm \ref{alg1}. We will then use
Theorem \ref{th:first} to derive sufficient conditions on the total training power and its distribution over the adaptive stages of Algorithm \ref{alg1} to guarantee estimating
the AoA/AoD of the mmWave channel with a desired resolution, and a certain bound on the maximum error probability.

\begin{theorem}
Algorithm \ref{alg1} succeeds in estimating the correct AoA and AoD of the single-path channel model in \eqref{eq:channel_model}, for a desired resolution $\frac{2 \pi}{N}$, with
an average probability of error $\bar{p}$ which is upper bounded by
\begin{equation}
\bar{p} \leq \frac{K^2-1}{2} \sum_{s=1}^\mathrm{S} \left(1-\frac{\left(1-\frac{1}{{{\beta}}_{s}}\right) P_{(s)} {G}^\mathrm{F}_{s} \bar{\gamma}}{4 \sqrt{1+\frac{1}{2}\left(1+\frac{1}{{{\beta}}_{s}} \right)P_{(s)} {G}^\mathrm{F}_{s} \bar{\gamma}+\frac{1}{16}P_{s}^2 {{G}^\mathrm{F}}_{s}^2 \bar{\gamma}^2 \left(1-\frac{1}{{{\beta}}_{s}}\right)^2}}\right),\label{eq:Pe}
\end{equation}
where ${{\beta}}_{s}=\frac{{G}^F_{s}}{{G}^B_{s}}=\min_{\substack{ \forall k_\mathrm{BS},k_\mathrm{MS}=1,2,..., K^{s-1} \\ \forall d=1,2,..., N^2 \\ \forall \bar{m} \in \cG_{\bar{m}}}}{\beta_{(s,k_\mathrm{BS},k_\mathrm{MS},d,\bar{m})}}$, ${G}^\mathrm{F}_{s}$ is the corresponding forward beamforming gain, and $\bar{\gamma}$ is the average channel SNR defined as $\bar{\gamma}=\frac{\bar{P}_R}{\rho \sigma^2}.$
\label{th:first}
\end{theorem}

\begin{proof}
If the BS and MS use Algorithm \ref{alg1} to estimate their AoA/AoD with a resolution $\frac{2 \pi}{N}$, and employ $K$ precoding and measurement vectors of codebooks $\cF$ and $\cW$ at each stage, then the output of the first stage can be written as \eqref{eq:sparse4}
\begin{align}
\by_{(1)}&=\sqrt{P_{(1)}}\left( \sqrt{\frac{G_{(1)}}{N_\mathrm{BS} N_\mathrm{MS}}}\bG_{(1,1,1)}+\bE_{(1,1,1)} \right)\bz+\bn_\mathrm{1}\\
&=\sqrt{P_{(1)}}\left[ \begin{array} {ccc} \sqrt{\frac{G_{(1)}}{N_\mathrm{BS} N_\mathrm{MS}}} \sum_{i=1}^{N^2/K^2}{\left[x\right]_i}+\sum_{i=1}^{N^2}{\left[\bE_{(1,1,1)}\right]_{1,i} \left[x\right]_i}\\ \vdots \\ \sqrt{\frac{G_{(1)}}{N_\mathrm{BS} N_\mathrm{MS}}} \sum_{i=(K^2-1)\frac{N^2}{K^2}+1}^{N^2}{\left[x\right]_i}+\sum_{i=1}^{N^2}{\left[\bE_{(1,1,1)}\right]_{K^2,i} \left[x\right]_i}\end{array}\right]+\bn_\mathrm{1}.
\end{align}

Without loss of generality, if we assume that the single non-zero element of $\bz$ is in the first location, then using the definition of the vector $\boldsymbol\alpha$ in \eqref{eq:channel2}, we get
\begin{align}
\by_{(1)}=\left[ \begin{array} {cccc} \sqrt{\frac{P_{(1)} N_\mathrm{BS} N_\mathrm{MS}}{\rho}}\left(\sqrt{\frac{G_{(1)}}{N_\mathrm{BS} N_\mathrm{MS}}}+\left[\bE_{(1,1,1)}\right]_{1,1} \right)\alpha + n_1\\ \sqrt{\frac{P_{(1)} N_\mathrm{BS} N_\mathrm{MS}}{\rho}} \left[\bE_{(1,1,1)}\right]_{2,1} \alpha + n_2 \\ \vdots \\ \sqrt{\frac{P_{(1)} N_\mathrm{BS} N_\mathrm{MS}}{\rho}} \left[\bE_{(1,1,1)}\right]_{K^2,1} \alpha + n_{K^2} \end{array} \right].
\end{align}

To select the partition of $\bz$ with the highest probability to carry the non-zero element, Algorithm \ref{alg1} chooses the partition with the maximum received power. Hence, the probability of successfully estimating the correct AoA/AoD range at this stage is the probability of the event $\bigcap_{r=1}^{K^2}\left\{\left[\by_{(1)}\right]_{1}^2 > \left[\by_{(1)}\right]_{r}^2\right\}$. Taking the complement of this event, and using the union bound, we write the probability of error at stage $s$ conditioned on the channel gain $p_{(s)}(\alpha)$ as
\begin{align}
p_{(s)}(\alpha)&=\mathrm{P}\left(\left.\bigcup_{r=1}^{K^2}{\left\{\left[\by_{(1)}\right]_{1}^2 < \left[\by_{(1)}\right]_{r}^2\right\}} \right| \alpha\right)\\
&\leq \sum_{r=1}^{K^2} \mathrm{P}\left(\left.\left[\by_{(1)}\right]_{1}^2 < \left[\by_{(1)}\right]_{r}^2\right|\alpha\right).\label{eq:Ps}
\end{align}

Now, note that $\left[\by_{(1)}\right]_{1} \sim \cN\left(\mu_1, \sigma^2\right)$ with $\mu_1=\sqrt{\frac{P_{(1)} N_\mathrm{BS} N_\mathrm{MS}}{\rho}}\left(\sqrt{\frac{G_{(1)}}{N_\mathrm{BS} N_\mathrm{MS}}}+\left[\bE_{(1,1,1)}\right]_{1,1} \right)\alpha$, and $\left[\by_{(1)}\right]_{r} \sim \cN\left(\mu_r, \sigma^2\right)$ with and $\mu_r=\sqrt{\frac{P_{(1)} N_\mathrm{BS} N_\mathrm{MS}}{\rho}} \left[\bE_{(1,1,1)}\right]_{r,1} \alpha , r=2,3, ..., K^2$. Using the result of \cite{proakisdigital} for the probability that the difference between the magnitude squares of two Gaussian random variables is less than zero, we reach
\begin{equation}
\mathrm{P}\left(\left.\left[\by_{(1)}\right]_{1}^2 < \left[\by_{(1)}\right]_{r}^2\right|\alpha\right)=\mathrm{Q}_1(a,b)-\frac{1}{2} \mathrm{I}_0(ab)\exp\left(-\frac{1}{2}\left(a^2+b^2\right)\right),
\end{equation}
with $a=\frac{|\mu_r|}{\sqrt{2\sigma^2}}=\sqrt{\frac{P_{(1)} G^\mathrm{F}_{(1,1,1,1)}}{2 \rho \sigma^2}}|\alpha|$, and $b=\frac{|\mu_1|}{\sqrt{2\sigma^2}}=\sqrt{\frac{P_{(1)} G^\mathrm{B}_{(1,1,1,1,r)}}{2 \rho \sigma^2}}|\alpha|$. where $\mathrm{Q}_1$ is the first-order Marcum $\mathrm{Q}$-function, and $\mathrm{I}_0$ is the $0$th order modified Bessel function of the first kind. In \cite{simon1998new,simon1998unified}, a new integral form of the $\mathrm{Q}$-function was derived, by which we get
\begin{align}
\begin{split}
\mathrm{P}\left(\left.\left[\by_{(1)}\right]_{1}^2 < \left[\by_{(1)}\right]_{r}^2\right|\alpha\right)=\frac{1}{4 \pi} & \int_{- \pi}^{\pi}
\frac{1-\frac{1}{\beta_{(1,1,1,1,r)}}}{1+2 \sin(\phi) \sqrt{\frac{1}{\beta_{(1,1,1,1,r)}}}+\frac{1}{\beta_{(1,1,1,1,r)}}} \\
& \hspace{-15pt}\times \exp\left(- \frac{P_{(1)} G^\mathrm{F}_{(1,1,1,1)} |\alpha|}{4 \rho \sigma^2} \left({1+2 \sin(\phi)\sqrt{\frac{1}{\beta_{(1,1,1,1,r)}}}+\frac{1}{\beta_{(1,1,1,1,r)}}}\right) \right) d\phi.
\end{split}\label{eq:Prob1}
\end{align}

We can now substitute by \eqref{eq:Prob1} in \eqref{eq:Ps} to obtain
\begin{align}
\begin{split}
p_{(s)}(\alpha) \leq \frac{1}{4 \pi} & \int_{- \pi}^{\pi} \sum_{r=1}^{K^2}
\frac{1-\frac{1}{\beta_{(1,1,1,1,r)}}}{1+2 \sin(\phi) \sqrt{\frac{1}{\beta_{(1,1,1,1,r)}}}+\frac{1}{\beta_{(1,1,1,1,r)}}} \\
& \hspace{30pt} \times \exp\left(- \frac{P_{(1)} G^\mathrm{F}_{(1,1,1,1)} |\alpha|}{4 \rho \sigma^2} \left({1+2 \sin(\phi)\sqrt{\frac{1}{\beta_{(1,1,1,1,r)}}}+\frac{1}{\beta_{(1,1,1,1,r)}}}\right) \right) d\phi.
\end{split}\label{eq:Prob2}
\end{align}

From \eqref{eq:Prob2}, we can show that $\frac{\partial p_{(s)}(\alpha)}{\partial \beta_{(1,1,1,1,r)}} < 0$ for $ 1 > \beta_{(1,1,1,1,r)} \geq 0$. This is expected as $\beta_{(1,1,1,1,r)}$ represents the forward to backward gain which is intuitively negatively proportional with the probability of error. Hence, we can bound $p_{(s)} (\alpha)$ as
\begin{align}
\begin{split}
p_{(s)}(\alpha) \leq \frac{K^2-1}{4 \pi} & \int_{- \pi}^{\pi}
\frac{1-\frac{1}{\beta_1}}{1+2 \sin(\phi) \sqrt{\frac{1}{\beta_1}}+\frac{1}{\beta_1}} \\
& \hspace{50pt} \times \exp\left(- \frac{P_{(1)} G^\mathrm{F}_1 |\alpha|}{4 \rho \sigma^2} \left({1+2 \sin(\phi)\sqrt{\frac{1}{\beta_1}}+\frac{1}{\beta_1}}\right) \right) d\phi,
\end{split}\label{eq:Prob3}
\end{align}
where ${\beta}_{1}=\min_{\substack{ \forall d=1,2,..., N^2 \\ \bar{m}\in \cG_{\bar{m}}}}{\beta_{(1,1,1,d,\bar{m})}}$, and $G^\mathrm{F}_1$ is the corresponding forward beamforming gain.

Using a similar analysis for each stage $s$, the total probability of error conditioned on the path gain can be now defined, and bounded again using the union bound as
\begin{align}
\begin{split}
p(\alpha)& =\mathrm{P}\left(\bigcup_{s=1}^{\mathrm{S}}\left(\left.\bigcup_{r=2}^{K^2} \left[\by_{(s)}\right]^2_1 < \left[\by_{(s)}\right]^2_r \right| \alpha\right)\right) \\
& \leq \frac{K^2-1}{4 \pi} \sum_{s=1}^{S} \int_{- \pi}^{\pi}
\frac{1-\frac{1}{\beta_s}}{1+2 \sin(\phi) \sqrt{\frac{1}{\beta_s}}+\frac{1}{\beta_s}} \\
& \hspace{100pt} \times \exp\left(- \frac{P_{(s)} G^\mathrm{F}_s |\alpha|}{4 \rho \sigma^2} \left({1+2 \sin(\phi)\sqrt{\frac{1}{\beta_s}}+\frac{1}{\beta_s}}\right) \right) d\phi,
\end{split}
\end{align}
where ${\beta}_{s}=\frac{G^F_s}{G^B_s}=\min_{\substack{ \forall k_\mathrm{BS},k_\mathrm{MS}=1,2,..., K^{s-1} \\ \forall d=1,2,..., N^2 \\ \bar{m}\in \cG_{\bar{m}}}}{\beta_{(s,k_\mathrm{BS},k_\mathrm{MS},d,\bar{m})}}$, and $G^\mathrm{F}_s$ is the corresponding forward beamforming gain.

Finally, to obtain the average probability of error $\bar{p}$, we need to integrate over the exponential distribution of $|\alpha|^2$. However, by swapping the summation with the integration sign, we will get again an integral similar to that given and solved in equations (27)-(35) of \cite{simon1998unified}, and by which we can directly obtain the bound in \eqref{eq:Pe}
\end{proof}

For the case when $\beta_s \to \infty$, i.e., when the backward gain in negligible and $\bE_{(s,k_\mathrm{BS},k_\mathrm{MS})} \to \boldsymbol{0}$, we can proceed further, and obtain a sufficient condition on the training power distribution to guarantee estimating the AoA/AoD of the channel with a certain bound on the maximum probability of error.

\begin{corollary}
 Consider using Algorithm \ref{alg1} to estimate the AoA and AoD of the single-path mmWave channel of model \eqref{eq:channel_model}, with a resolution $\frac{2 \pi}{N}$, with $K$ precoding and measurement vectors of $\cF, \cW$ used at each stage, and with $\beta_s \to \infty$, and $\bE_{s,k_\mathrm{BS},\mathrm{MS}} \to 0$. If the power at each stage $P_{(s)}, s=1,2,..., \mathrm{S}$ satisfies:
\begin{equation}
P_{(s)} \geq \frac{\Gamma}{G_{(s)}}
\end{equation}
with
\begin{equation}
\Gamma=\frac{2}{\bar{\gamma}}\left(\frac{(K^2-1)\mathrm{S}}{\delta}-2\right),
\end{equation}
then, the AoA and AoD are guaranteed to be estimated with an average probability of error $\bar{p} \leq \delta $.
\label{cor:Cor2}
\end{corollary}
To prove  Corollary \ref{cor:Cor2}, it is sufficient to substitute with the given $P_{(s)}$, and $\Gamma$ in \eqref{eq:Pe} to get $\bar{p} \leq \delta$.

Also, from Corollary \ref{cor:Cor2}, it is easy to show that a total training power $P_\mathrm{T}$, with $P_\mathrm{T} \geq K^2 \Gamma \sum_{s=1}^{\mathrm{S}}{\frac{1}{G_{(s)}}}$ is sufficient to estimate the AoA/AoD of the single-path mmWave channel with $\bar{p} \leq \delta $ if it is distributed according to the way described in Corollary \ref{cor:Cor2}.

Finally, if we have a bound on the total training power, we can use Theorem \ref{th:first} to get an upper bound on the error probability.

\begin{corollary}
 Consider using Algorithm \ref{alg1} to estimate the AoA and AoD of the single-path mmWave channel of model \eqref{eq:channel_model}, with a resolution $\frac{2 \pi}{N}$, with $K$ precoding and measurement vectors of $\cF, \cW$ used at each stage, and with $\beta_s \to \infty$, and $\bE_{s,k_\mathrm{BS},\mathrm{MS}} \to 0$. If the total training power is $P_\mathrm{T}$, and if this power is distributed over the adaptive stages of Algorithm \ref{alg1} such that:
\begin{equation}
P_{(s)} = \frac{P_\mathrm{T}}{K^2 \sum_{n=1}^{\mathrm{S}}{\frac{G_{(s)}}{G_{(n)}}}}, s=1,2,..., \mathrm{S}
\end{equation}

Then, the AoA and AoD are guaranteed to be estimated with an average probability of error $\bar{p}$ where
\begin{equation}
\bar{p}\leq \frac{(K^2-1)S}{\frac{P_\mathrm{T} \bar{\gamma}}{2 K^2 \sum_{s=1}^{\mathrm{S}}{\frac{1}{G_{(s)}}}}+2}.
\end{equation}
\label{cor:Cor3}
\end{corollary}
To prove  Corollary \ref{cor:Cor3}, it is sufficient to substitute with the given $P_{(s)}$ in \eqref{eq:Pe} to get the bound on $\bar{p}$.

In this section, the main idea of the proposed adaptive mmWave channel estimation algorithm was explained and analyzed for the single-path channels. Now, we extend this algorithm to the general case of multi-path mmWave channels.

\subsection{Adaptive Channel Estimation Algorithm for Multi-Path MmWave Channels}\label{subsec:multi-path}

Consider the case when multiple paths exist between the BS and MS. Thanks to the poor scattering nature of the mmWave channels, the channel estimation problem can be formulated as a sparse compressed sensing problem as discussed in \sref{sec:Prob_Form}. Consequently, a modified matching pursuit algorithm can be used to estimate the AoAs and AoDs along with the corresponding path gains of $L_\mathrm{d}$ paths of the channel, where $L_\mathrm{d}$ is the number of dominant paths need to be resolved. Given the problem formulation in \eqref{eq:sparse_formulation2}, the objective now is to determine the $L_\mathrm{d}$ non-zero elements of $\bz$ with the maximum power. Based on the single-path case, we propose Algorithm~\ref{alg2} to adaptively estimate the different channel parameters.
\begin{algorithm} [!t]                     
\caption{Adaptive Estimation Algorithm for Multi-Path MmWave Channels}          
\label{alg2}                                                      
\begin{algorithmic}
    \State \textbf{Input:} BS and MS know $N, K, L_d$, and have $\cF, \cW$
    \State \textbf{Initialization:} $\bT^\mathrm{BS}_{(1,1)}=\{1, ..., 1\}, \bT^\mathrm{MS}_{(1,1)}=\{1, ..., 1\}, \mathrm{S}=\log_{K}{\left(N/L_d\right)}$                                              
    \For{$\ell \leq L_d$}
        \For{$s \leq \mathrm{S}$}
            \For {$m_\mathrm{BS} \leq K L_d$}
                \State BS transmits a training symbol using $\left[\bF_{(s,\bT^\mathrm{BS}_{(\ell,s)})}\right]_{:,m_\mathrm{BS}}$
                \For {$m_\mathrm{MS} \leq K L_d$}
                      \State MS makes a measurement using $\left[\bW_{(s,\bT^\mathrm{BS}_{(\ell,s)})}\right]_{:,m_\mathrm{MS}}$
                \EndFor
            \State After MS measurements: $\by_{m_\mathrm{BS}}= \sqrt{P_{s}}\left[\bW_{(s,\bT^\mathrm{MS}_{(\ell,s)})}\right] \bH \left[\bF_{(s,\bT^\mathrm{BS}_{(\ell,s)})}\right]_{:,m_\mathrm{BS}}+\bn_{m_\mathrm{BS}}$
            \EndFor
            \State $\by_{(s)}=[\by_1^T, \by_2^T, ..., \by_K^T]^T$
            \For {$p=1 \leq \ell-1$} Project out the contributions of the previously estimated paths
                 \State $\bg=\bF^T_{(s,\bT^\mathrm{BS}_{(p,s)})} \left[\bA_\mathrm{BS,D}\right]^*_{:,{\bT^\mathrm{BS}_{(p,s)}(1)}} \otimes \bW^H_{(s,\bT^\mathrm{MS}_{(p,s)})} \left[\bA_\mathrm{MS,D}\right]_{:,\bT^\mathrm{MS}_{(p,s)}(1)}$
                \State $\by_{(s)}=\by_{(s)}-\by_{(s)}^H \bg \left(\bg^H \bg\right) \bg$
            \EndFor
            \State $\bY=\mathrm{matix}(\by_{(s)})$ Return $\by_{(s)}$ to the matrix form
            \State $\left(m_\mathrm{BS}^{\star},m_\mathrm{MS}^{\star}\right)=\arg\max_{ \forall m_\mathrm{BS}, m_\mathrm{MS}=1,2,..., K} \left[\bY \odot \bY^* \right]_{m_\mathrm{MS},m_\mathrm{BS}}$
            \State $\bT^\mathrm{BS}_{(\ell,s+1)}(1)=K(m_\mathrm{BS}^\star-1)+1, \bT^\mathrm{MS}_{(\ell,s+1)}(1)=K(m_\mathrm{MS}^\star-1)+1$
            \For {$p=1 \leq \ell-1$}
                \State $\bT^\mathrm{BS}_{(\ell,s+1)}(p)=\bT^\mathrm{BS}_{(p,s+1)}(1), \bT^\mathrm{MS}_{(\ell,s+1)}(p)=\bT^\mathrm{MS}_{(p,s+1)}(1)$
            \EndFor
        \EndFor
        \State $\hat{\phi}_\ell=\bar{\phi}_{\bT^\mathrm{BS}_{(\ell,\mathrm{S}+1)}(1)}, \hat{\theta}_\ell=\bar{\theta}_{\bT^\mathrm{MS}_{(\ell,\mathrm{S}+1)}(1)}$
        \State $\bg=\bF^T_{(\mathrm{S},\bT^\mathrm{BS}_{(\ell,\mathrm{S})})} \left[\bA^*_\mathrm{BS,D}\right]_{:,\bT^\mathrm{BS}_{(\ell,\mathrm{S}+1)}(1)} \otimes \bW^H_{(\mathrm{S},\bT^\mathrm{MS}_{(\ell,\mathrm{S})})} \left[\bA_\mathrm{MS,D}\right]_{:,\bT^\mathrm{MS}_{(\ell,\mathrm{S}+1)}(1)}$
        \State $\hat{\alpha}_\ell=\sqrt{\frac{\rho}{P_{(\mathrm{S})}G_\mathrm{(S)}}}{\frac{\by_{(\mathrm{S})}^H \bg}{\bg^H \bg}  }$
    \EndFor
    \end{algorithmic}

\end{algorithm}

\textbf{Modified Hierarchical Codebook}: For the multi-path case, we need to make a small modification to the structure of the hierarchical codebook described in \sref{sec:codebook}.
 As will be explained shortly, the adaptive algorithm in the multi-path case starts by using $K L_d$ precoding and measurement vectors at the BS and MS instead of $K$. In each stage, $L_d$ of
 those $K L_d$ partitions are selected for further refinement by dividing each one into $K$ smaller partitions in the next stage. Hence, to take this into account, the first level of the codebook $\cF$ in \sref{sec:codebook} consists of one subset with $K L_d$ beamforming vectors that divide the initial AoD range into $K L_d$ ranges. Similarly, in each level $s, s>1$, the codebook $\cF_s$ has $K^{s-1} L_d$ levels, and the ranges $\cI_{(s,k)}$, and $\cI_{(s,k,m)}$ are consequently defined as $\cI_{(s,k)}=\left\{\frac{(k-1)N}{L_d K^{s-1}},...,\frac{kN}{L_d K^{s-1}}\right\}$, and $\cI_{(k,s,m)}=\left\{\frac{N}{L_d K^{s}}\left(K(k-1)+m-1\right)+1,..., \frac{N}{L_d K^{s}}\left(K(k-1)+m\right)\right\}$. Given these definitions of the quantized AoD ranges associated with each beamforming vector $m$, of the subset $k$, of level $s$, the design of the beamforming vectors proceeds identical to that described in \sref{subsec:codebook_design}.

To estimate the $L_d$ dominant paths of the mmWave channel, Algorithm \ref{alg2} makes $L_d$ outer iterations. In each one, an algorithm similar to Algorithm \ref{alg1} is executed to detect one more path after subtracting the contributions of the  previously estimated paths. More specifically, Algorithm \ref{alg2} operates as follows: In the initial stage, both the BS and MS use $K L_d $ beamforming vectors defined by the codebooks in \sref{sec:codebook} to divide the AoA, and AoD range into $KL_d$ sub-ranges each. Similar to the single-path case, the algorithm proceeds by selecting the maximum received signal power to determine the $L_d$ most promising sections to carry the dominant paths of the channel. This process is repeated until we reach the required AoD resolution, and only one path is estimated at this iteration. The trajectories used by the BS to detect the first path is stored in the matrix $\bT^\mathrm{BS}$  to be used in the later iterations. In the next iteration, a similar BS-MS precoding/measurement step is repeated. However, at each stage $s$, the contribution of the first path that has been already estimated in the previous iteration, which is stored in $\bT^\mathrm{BS}$, is projected out before determining the new promising AoD ranges. In the next stage $s+1$, two AoD ranges are selected for further refinement, namely, the one selected at stage $s$ of this iteration, and the one selected by the first path at stage $s+1$ of the first iteration which is stored in $\bT^\mathrm{BS}$. The selection of those two AoD ranges enables the algorithm to detect different path with AoDs separated by a resolution up to $\frac{2 \pi}{N}$. The algorithm proceeds in the same way until the $L_d$ paths are solved. After estimating the AoAs/AoDs with the desired resolution, the algorithm finally calculates the estimated path gains using a linear least squares estimator (LLSE).

Note that one disadvantage of the adaptive beamwidth algorithm in the multi-path case is the possible destructive interference between the path gains when they are summed up in the earlier stages of the algorithm. This disadvantage does not appear in the exhaustive search training algorithms; as only high resolution beams are used in estimating the dominant paths of the channel. The impact of this advantage on the operation of the proposed algorithm, however, is smaller in the case of mmWave channels thanks to the sparse nature of the channel.

The total number of adaptive stages required by Algorithm \ref{alg2} to estimate the AoAs/AoDs of the $L_\mathrm{d}$ paths of the channel with a resolution $\frac{2 \pi}{N}$ is $\log_K\left(\frac{N}{L_\mathrm{d}}\right)$. Since we need $K L_\mathrm{d}$ precoding vectors, and $K L_\mathrm{d}$ measurement vectors for each precoding direction in each stage, and since these adaptive stages are repeated for each path, the total number of steps required to estimate $L_d$ paths of the mmWave channel using the proposed algorithm is $K^2 L_\mathrm{d}^3 \log_K\left(\frac{N}{L_\mathrm{d}}\right)$. If multiple RF chains are used in the MS to combine the measurements, the required number of training time slots is then reduced to be $K L_d^2 \lceil\frac{K L_d}{\Nc}\rceil\log_K\left(\frac{N}{L_\mathrm{d}}\right)$.

\section{Hybrid Precoding Design} \label{sec:Design}


We seek now to design the hybrid precoders/combiners, ($\Frf$, $\Fbb$, $\Wrf$, $\Wbb$), at both the BS and MS to maximize the mutual information achieved with Gaussian signaling over the mmWave link in (\ref{eq:received_signal2})~\cite{goldsmith2003capacity} while taking the different RF precoding constraints into consideration. Regardless of whether uplink or downlink transmission is considered, the hybrid precoding problem can be summarized as directly maximizing the rate expression
\begin{align}
R =\log_2 \left|I_{N_\mathrm{S}}+\frac{P}{N_\mathrm{S}}\bR_\mathrm{n}^{-1}  {\Wbb}^H {\Wrf}^H \bH\Frf\Fbb
{\Fbb}^H {\Frf}^H { \bH}^H \Wrf \Wbb \right|,
 \label{eq:rate1}
\end{align}
over the choice of feasible analog and digital processing matrices  ($\Frf$, $\Fbb$, $\Wrf$, $\Wbb$). Note that in (\ref{eq:rate1}), $\bR_\mathrm{n}$ is the post-processing noise covariance matrix, i.e.,  $\bR_\mathrm{n}={\Wbb}^H {\Wrf}^H {\Wrf} {\Wbb}$ in the downlink, and $\bR_\mathrm{n}= {\Fbb}^H {\Frf}^H {\Frf} {\Fbb}$ in the uplink.

For simplicity of exposition, we begin in this section by summarizing the process with which the BS calculates the hybrid precoding matrices, ($\Frf$, $\Fbb$), to be used on the downlink. Calculation of the uplink precoders used by the MS follows in an identical manner.

We propose to split the precoding problem into two phases. In the first phase, the BS and MS apply the adaptive channel estimation algorithm of \sref{sec:Prob_Form} to estimate the mmWave channel parameters. At the end of the channel training/estimation phase, the BS constructs the downlink channel's matrix leveraging the geometric structure of the channel. If the channel is not reciprocal, the estimation algorithm of \sref{sec:Algorithm} can be used to construct the uplink channel matrix at the MS. At this stage, the MS leverages the basis pursuit algorithm in \cite{ayach2013spatially} to compute $\Wrf$ and $\Wbb$ so that their combined effect, $\Wrf\Wbb$, approximates the dominant eigenvectors of the uplink's channel.

As a result of the downlink channel training/estimation phase in \sref{sec:Algorithm}, the BS now has estimated knowledge of its own steering matrix $\hat{\bA}_\mathrm{BS}$, the MS steering matrix $\hat{\bA}_\mathrm{MS}$, and the estimated path gain vector $\hat{\boldsymbol\alpha}$. Thus, the BS may construct the estimated downlink channel matrix as
\begin{equation}
\hat{\bH}=\hat{\bA}_\mathrm{MS} \mathrm{diag}\left(\hat{\boldsymbol\alpha}\right) \hat{\bA}_\mathrm{BS}^H. \label{eq:channelf}
\end{equation}

The BS can now build its hybrid data precoders $\Frf$ and $\Fbb$ to approximate the dominant singular vectors of the channel, $\hat{\bH}$, denoted by the unconstrained precoder $\bF_\mathrm{opt}$.

At this stage, we recall that the precoding capability of the system in Fig. \ref{fig:BS_MU_arch_Fig} can be summarized as the ability to apply a set of $\Nc$ constrained analog beamforming vectors, via the analog precoder $\Frf$, and form a linear combination of them via its digital precoder $\Fbb$. Following the methodology in \cite{ayach2013spatially},
the problem of approximating the unconstrained precoder $\bF_\mathrm{opt}$ can be written as~\cite{ayach2013spatially}
\begin{align}
\begin{split}
(\bF_\mathrm{RF}^{\star}, \bF_\mathrm{BB}^{\star}) &  = \\  & \hspace{-17pt} \arg\min \ \  \|\bF_\mathrm{opt}-\Frf\Fbb\|_F, \\
& \hspace{-15pt} \mathrm{s.t}. \ \  \left[\Frf\right]_{:,i} \in \left\{\left[\bA_\mathrm{can}\right]_{:,\ell} | \ 1 \leq \ell \leq N_\mathrm{can}\right\}, i=1,2,..., N_\mathrm{RF} \\
& \hspace{9pt} \|\Frf\Fbb\|_F^2=N_\mathrm{S}.
\label{eqn:combinationprecoding}
\end{split}
\end{align}

This problem is similar to \eqref{eq:BF_Design}, and can be again formulated as a sparse optimization problem as in \eqref{eq:BF_Design2}. Finally, the BS baseband and RF precoders,
$\Fbb$ and $\Frf$, can be designed using the iterative matching pursuit procedure in Algorithm \ref{alg3}, but instead of $\bff_\mathrm{res}$, we define the matrix $\bF_\mathrm{res}=\bF_\mathrm{opt}$,
and instead of $\bff^\star$, we define $\bF^\star=\bF_\mathrm{opt}$.

\section{Simulation Results} \label{sec:Results}
In this section, we present numerical results to evaluate the performance of the proposed training codebook, adaptive channel estimation algorithm, and hybrid precoding algorithm. We firstly consider a single BS-MS link, and then show some results for the mmWave cellular channel model.

\subsection{Performance Evaluation with Point-to-Point Channels} \label{subsec:link}
In these simulations, we consider the case when there is only one BS and one MS, i.e., without any interference. The system model and the simulation scenario are as follows:

\textbf{System Model}  We adopt the hybrid analog/digital system architecture presented in \figref{fig:BS_MU_arch_Fig}. The BS has $N_\mathrm{BS}=64$ antennas, and $10$ RF chains, the MS has $N_\mathrm{MS}=32$ antennas and $6$ RF chains. The antenna arrays are ULAs, with spacing between antennas equal to $\lambda/2$, and the RF phase shifters are assumed to have only quantized phases. Hence, only a finite set of the RF beamforming vectors is allowed, and assumed to be beamsteering vectors, as discussed in \sref{subsec:codebook_design}, with $7$ quantization bits.

\textbf{Channel Model} We consider the channel model described in \eqref{eq:channel_model}, with $\bar{P}_R=1$, and a number of paths $L=3$. The AoAs/AoDs are assumed to take continuous values, i.e., not quantized, and are uniformly distributed in the range $\left[0, 2 \pi\right]$. The system is assumed to operate at $28 \mathrm{GHz}$ carrier frequency, has a bandwidth of $100 \mathrm{MHz}$, and the path-loss exponent equals $n_\mathrm{pl}=3$.

\textbf{Simulation Scenario} All the simulations in this section will present spectral efficiency results with different system, and algorithms parameters. To generate these results, the channel parameters are estimated using the algorithms presented in \sref{sec:Algorithm}, which in turn use the hierarchical training codebooks designed in \sref{sec:codebook}. After estimating its parameters, the geometrical channel is reconstructed according to \eqref{eq:channelf}, and is used in the design of the hybrid precoders and decoders according to \sref{sec:Design}. Unless otherwise mentioned, these are the parameters used for both of the two steps:
\begin{enumerate}
\item{Channel estimation parameters: For the single-path channels, Algorithm \ref{alg1} is used to estimate the channel parameters with AoA/AoD resolution parameter $N=64$, and with $K=2$ beamforming vectors at each stage. For the multi-path case, the parameters $N, K, L_d$ will be defined with each simulation. The training power are determined according to Corollary $\ref{cor:Cor2}$, with a desired maximum probability of error $\delta = 0.05$. Hence, the training power changes based on the parameter $K$, and $N$. Also, the total training power is distributed over the adaptive estimation stages according to Corollary \ref{cor:Cor2}.}
\item{Hybrid precoding parameters: The hybrid precoding matrices are constructed with the same available system architecture described above, and assuming a number of multiplexed streams $N_\mathrm{S}=L_d$.
}
\end{enumerate}

\begin{figure}[t]
\centerline{
\includegraphics[width=5 in, height= .5\textwidth]{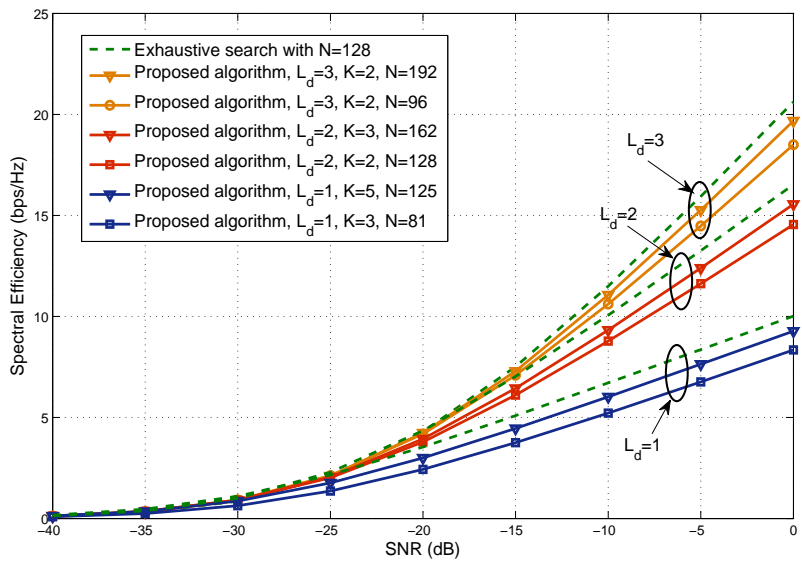}
}
\caption{Spectral efficiency achieved when the precoding matrices are built using the mmWave channel estimated by the proposed algorithms in a channel with $L=3$, and $L_\mathrm{d}=1,2,3$. The figure compares the performance of the algorithm when different values of the parameter $K$ are chosen. The results indicate that a very close performance to the exhaustive search case can be achieved with $K<<N$, which maps to much smaller numbers of iterations.}
\label{fig:Estimation}
\vspace{-10pt}
\end{figure}

In \figref{fig:Estimation}, the precoding gains given by the proposed mmWave channel estimation algorithms are simulated for the cases when the desired number of estimated paths $L_\mathrm{d}$ equals $1,2$, and $3$. Algorithm \ref{alg1}, and Algorithm \ref{alg2} are simulated for different values of $K$, and compared with the precoding gain of the exhaustive search solution. The results indicate that comparable gains can be achieved using the proposed algorithms despite their low-complexity, and the requirement of a much smaller number of iterations. For example, for $L_\mathrm{d}=3$, and $K=2$, although only $96 \ll N_\mathrm{BS} N_\mathrm{MS} = 2048$ training steps are required, the spectral efficiency performance degradation is less than $1$ bps/Hz compared with the exhaustive search solution that requires much more iterations.

\begin{figure}[t]
\centerline{
\includegraphics[width=5 in, height=.5\textwidth]{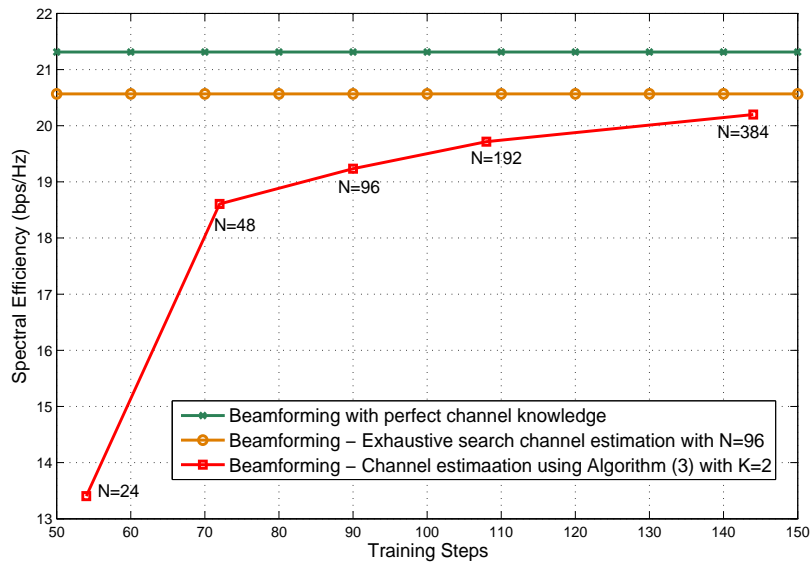}
}
\caption{The improvement of the spectral efficiency with the development of the adaptive channel estimation algorithm is shown and compared with the exhaustive search and perfect channel knowledge cases. While the exhaustive search in this case needs a large number of  iterations, a much smaller number of iterations may be sufficient to approximate its performance using the proposed adaptive algorithms.}
\label{fig:Training}
\vspace{-10pt}
\end{figure}

In \figref{fig:Training}, the improvement of the precoding gains achieved by the proposed algorithm for $L_\mathrm{d}=3$ with the training iterations is simulated. The results show that more than $90\%$ of the exhaustive search gain can be achieved with only $70$ iterations with $K=2$. These results also indicate that a wise choice of the desired resolution parameter $N$ is needed in order to have a good compromise between performance and training overhead. For example, the figure shows that doubling the number of training steps, i.e., from 70 to 140, achieves an improvement of only $1$ bps/Hz in the spectral efficiency.

\begin{figure}[t]
\centerline{
\includegraphics[width=5 in, height= .5\textwidth]{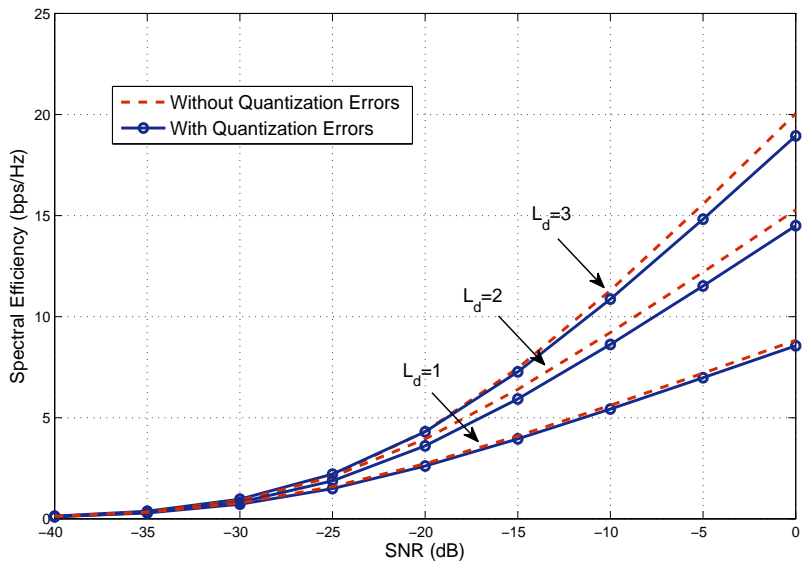}
}
\caption{The performance error due to the AoAs/AoDs quantization assumption in \eqref{eq:sparse_formulation} is evaluated. The performance error is the difference between the curve with continuous angles, and the one with quantization, as this continuity of angles' values is not taken into consideration while designing the algorithm.}
\label{fig:Quantization}
\vspace{-10pt}
\end{figure}

In \figref{fig:Quantization}, we evaluate the error in the performance of the proposed channel estimation algorithm caused by the AoAs/AoDs quantization assumption made in \eqref{eq:sparse_formulation}, the proposed algorithms are simulated for the cases when the channel AoAs/AoDs are quantized, i.e., when the used quantization assumption is exact, and when the AoAs/AoDs are continuous, i.e, with quantization error induced in our formulation. The figure plots the performance of the proposed algorithms for the cases $L_d=1, K=2, N=81$,$L_d=2$, $K=2, N=128$, and $L_d=3, K=3, N=96$, and show that the performance loss in our algorithms due to the quantization assumption is very small for large enough resolution parameters $N$.

\begin{figure}[t]
\centerline{
\includegraphics[width=5 in, height= .5\textwidth]{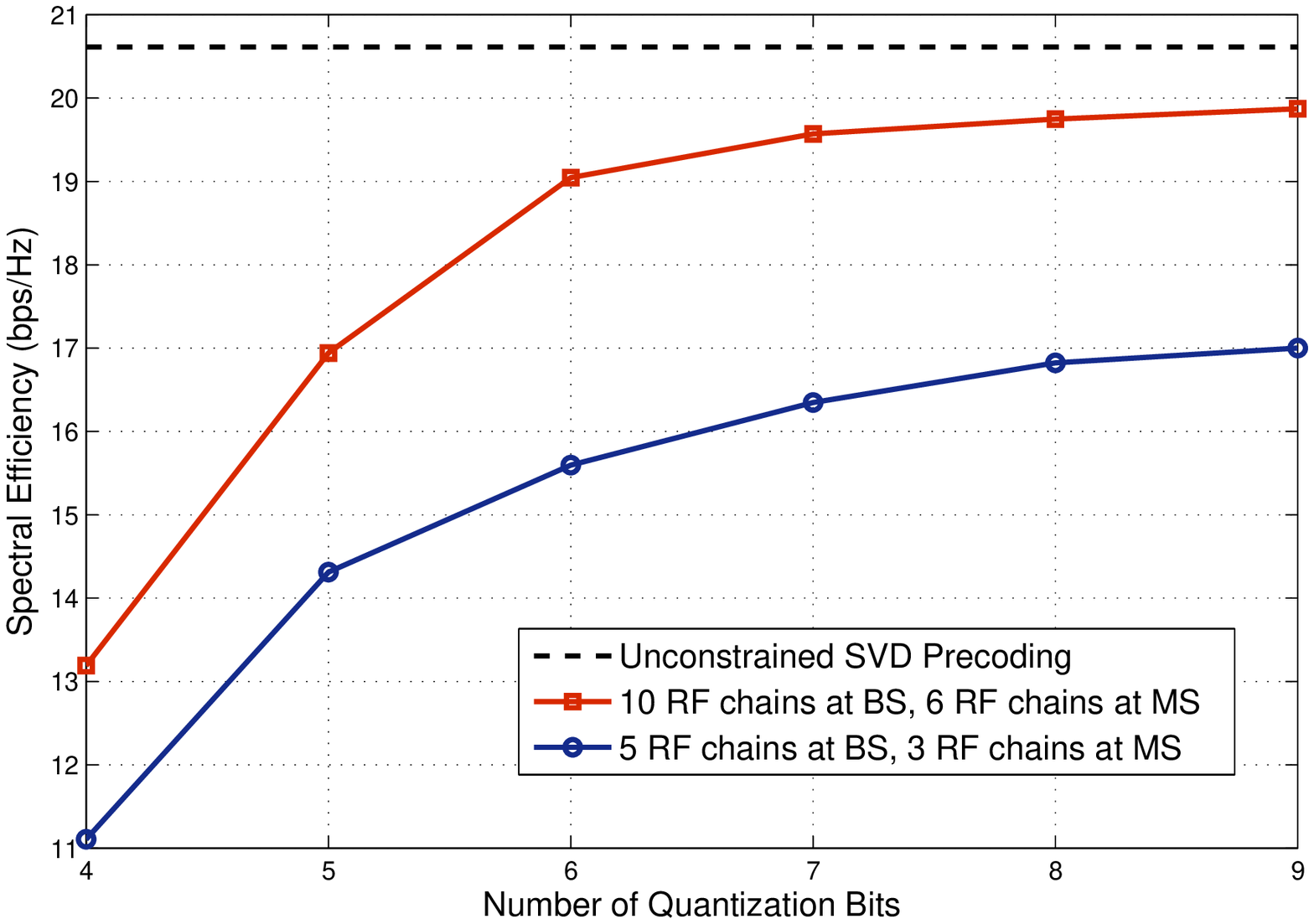}
}
\caption{Spectral efficiency as a function of phase quantization bits in a hybrid system with only quantized analog phase control. Results compare the performance of the hybrid analog digital channel estimation and precoding algorithms with the unconstrained digital system with perfect channel knowledge at an SNR of 0dB.}
\label{fig:RF}
\vspace{-10pt}
\end{figure}

In \figref{fig:RF}, the impact of the RF system limitations on the performance of the proposed channel estimation, and precoding algorithms, is evaluated, and compared with the case of constraints-free system. Two system models are considered in \figref{fig:RF}, one with 10 RF chains at the BS, 6 RF chains at the MS, and the other with  5 RF chains at the BS, and 3 RF chains at the MS. The other parameters are the same as the previous simulations with $L_d=3$. The performance achieved by those two systems is further simulated with different number of quantization bits of the phase shifters. Simulation results show that the proposed hybrid analog-digital precoding algorithm can achieve near-optimal data-rates compared with the unconstrained solutions if a sufficient number of RF chains, and quantization bits exist. Also, the results show that 5 quantization bits may be sufficient to accomplish more that $90\%$ of the maximum gain.

\subsection{Performance Evaluation with MmWave Cellular System Setup} \label{subsec:cellular_setup}

Now, we consider evaluating the proposed algorithm in a mmWave cellular system setting with out-of-cell interference. To provide a practical evaluation, we adopt the following stochastic geometry model.

\textbf{Network and System Models} The desired BS, in a cell of radius $R_c=100 m$, is assumed to communicate with a MS using the channel estimation, and hybrid precoding algorithms derived. Each MS is assumed to receive its desired signal $s_\mathrm{d}$ in addition to cellular interference. The interfering BSs follow a Poisson point process (PPP) $\Phi(\lambda)$ with $\lambda=\frac{1}{\pi R_c^2}$ to model the downlink out-of-cell interference \cite{andrews2011tractable,akoum2012coverage,Bai1}. To simulate a cellular setting, the nearest BS to the MS is always considered as the desired BS. The received signal at the MS can be then written as
\begin{align}
\by=\bW \bH_\mathrm{d} \bF_\mathrm{d} \bs_\mathrm{d} + \sum_{\substack{r_i \in \Phi(\lambda) \\ r_i \geq r_\mathrm{d} }}{\bW \bH_i \bF_i \bs_i}+\bn \label{eq:Int}
\end{align}
where $r_\mathrm{d}, r_i$ are the distances from the MS to the desired and the $i$th interfering BSs, respectively. Each interfering BS is assumed to have the same number of ULA antennas $N_\mathrm{BS}=64$, and to have the same horizontal orientation of the antenna arrays, i.e., all the beamforming is in the azimuth domain. Further, each BS generates a beamsteering beamforming vectors that steers its signal in a uniform random direction, i.e., $\bF_i=\ba_\mathrm{BS}\left(\phi_i\right)$, $\phi_i$ is uniformly chosen in $\left[0, 2 \pi\right]$. $\bH_i$ has the same definition in \eqref{eq:channel_model} with the path loss calculated for each BS based on its distance $r_i$. For fairness, all BSs are assumed to transmit with the same average power $P$. All the other system parameters are similar to the previous section.

In each stage $s$ of the estimation phase, the received signal at the MS is given by \eqref{eq:Int} with $\bW$ and $\bF$ equal to the BS and MS training precoders and combiners descried in \sref{sec:Algorithm}. Hence, the cellular interference affects the maximum power detection problem at every stage of the channel estimation algorithm. After the channel is estimated, the precoders $\bW$ and $\bF$ are designed as shown in \sref{sec:Design}.

To evaluate the performance of the proposed hybrid precoding algorithm, we adopt the coverage probability as a performance metric. As we are interested in multiplexing many streams per user, we define the coverage probability relative to the rate instead of the signal to interference and noise ratio (SINR). Consequently, we use the following definition of the coverage probability

\begin{equation}
P_{(c)}\left(\eta\right)=\mathrm{P}(R \geq \eta).
\end{equation}\label{eq:Coverage}

An outage happens if the user's rate falls below a certain threshold $\eta$.

\begin{figure}[t]
\centerline{
\includegraphics[width=5 in, height= .5\textwidth]{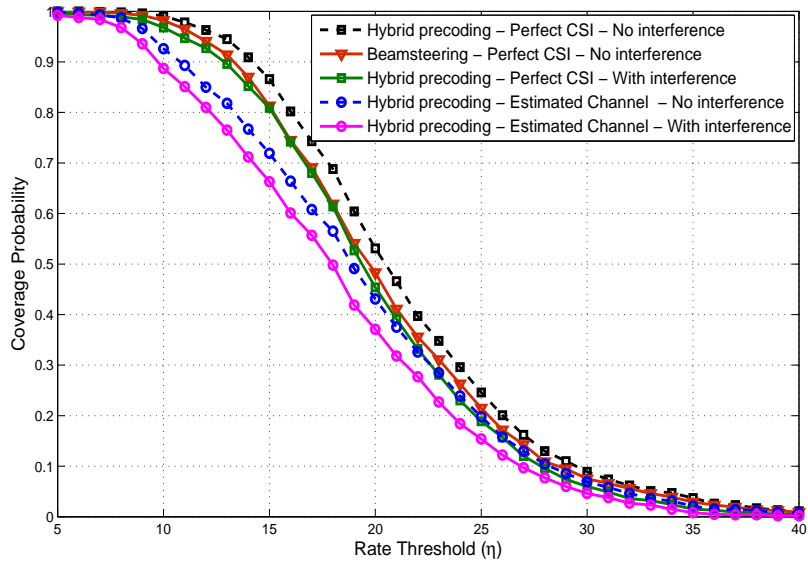}
}
\caption{Coverage probabilities of the proposed channel estimation and precoding algorithms in a mmWave cellular system setting with PPP interference. The figure compares the different cases when the estimation and/or interference error exist to evaluate the effect of each of them on the proposed algorithms.}
\label{fig:Coverage_HP}
\vspace{-10pt}
\end{figure}

\textbf{Scenario and Results} In \figref{fig:Coverage_HP}, the coverage probability is evaluated as described before. The curves with 'Estimated Channel' label represents the case when Algorithm (3) is used to estimate the channel parameters in the presence of interference. After estimating the channel, this interference is taken into consideration again in calculating the coverage probability in the curve labeled 'With Interference', and omitted for the curve with the label 'No Interference'. Hence, those two curves represent the cases when cellular interference affects both the channel estimation and data transmission phases, or the channel estimation phase only. The presented results compare the performance of the mentioned scenario using the proposed algorithms, with the case when the hybrid precoding algorithm in \sref{sec:Design} is designed based on perfect channel state information (CSI). They are also compared with the case when only analog beamforming is used to steer the signal towards the dominant channel paths. The results show that a reasonable gain can be achieved with the proposed hybrid precoding algorithm due to its higher capability of managing the inter-stream interference, in addition to overcoming the RF hardware constraints. The simulations also indicate that the effect of the cellular interference of the performance of the channel estimation and precoding algorithms is not critical despite of the low-complexity of the proposed algorithms.

\section{Conclusions}\label{sec:conclusion}
In this paper, we considered a single-user mmWave system setting, and investigated the design of suitable mmWave channel estimation and precoding algorithms. First, we formulated, and developed a hierarchical multi-resolution codebook based on hybrid analog/digital precoding. We then proposed mmWave channel estimation algorithms that efficiently detect the different parameters of the mmWave channel with a low training overhead. The proposed algorithms depend on the developed sparse formulation of the poor scattering mmWave channel, and on the designed hierarchical codebooks to adaptively estimate the channel parameters. The performance of the proposed algorithm is analytically evaluated for the single-path channel case, and some insights into efficient training power distributions are obtained. Despite the low-complexity, simulation results showed that the proposed channel estimation algorithm realizes spectral efficiency and precoding gain that are comparable to that obtained by exhaustive search. The mmWave hybrid precoding algorithms are also proved to achieve a near-optimal performance relative to the unconstrained digital solutions, and attain reasonable gains compared with analog-only beamforming. The attained precoding gains can be also stated in terms of the coverage probability of mmWave cellular systems. For future work, it would be interesting to consider mmWave channels with random blockage between the BS and MS \cite{bai2013analysis}, and seek the design of robust adaptive channel estimation algorithms. Besides the channel estimation algorithms developed in this paper assuming fixed and known array structures, it would be also important for mmWave systems to develop efficient algorithms that adaptively estimate the channel with random or time-varying array manifolds.

\begin{small}
\bibliographystyle{IEEEtran}

\bibliography{IEEEabrv,Heathabrv,Ahmed}

\end{small}

\end{document}

%% file: input.tex
\usepackage{amsfonts}
\usepackage{times}
\usepackage{graphicx}
\usepackage{latexsym}
\usepackage{dsfont}
\usepackage{amssymb}
\usepackage{amsmath}
\usepackage{cite}
\usepackage{verbatim}
\usepackage{subfigure}
\newtheorem{theorem}{Theorem}

\newtheorem{algorithm}[theorem]{Algorithm}

\newtheorem{corollary}[theorem]{Corollary}

\newenvironment{proof}{ \textbf{Proof:} }{ \hfill $\Box$}

\newcommand{\figref}[1]{{Fig.}~\ref{#1}}

\def\bb0{{\mathbb{0}}}

\def\ba{{\mathbf{a}}}
\def\bb{{\mathbf{b}}}

\def\bff{{\mathbf{f}}}
\def\bg{{\mathbf{g}}}

\def\bm{{\mathbf{m}}}
\def\bn{{\mathbf{n}}}

\def\br{{\mathbf{r}}}
\def\bs{{\mathbf{s}}}

\def\bw{{\mathbf{w}}}
\def\bx{{\mathbf{x}}}
\def\by{{\mathbf{y}}}
\def\bz{{\mathbf{z}}}
\def\b0{{\mathbf{0}}}

\def\bA{{\mathbf{A}}}
\def\bB{{\mathbf{B}}}

\def\bE{{\mathbf{E}}}
\def\bF{{\mathbf{F}}}
\def\bG{{\mathbf{G}}}
\def\bH{{\mathbf{H}}}
\def\bI{{\mathbf{I}}}

\def\bQ{{\mathbf{Q}}}
\def\bR{{\mathbf{R}}}
\def\bS{{\mathbf{S}}}
\def\bT{{\mathbf{T}}}

\def\bW{{\mathbf{W}}}

\def\bY{{\mathbf{Y}}}

\def\bbE{{\mathbb{E}}}

\def\cA{\mathcal{A}}

\def\cF{\mathcal{F}}
\def\cG{\mathcal{G}}

\def\cI{\mathcal{I}}

\def\cN{\mathcal{N}}
\def\cO{\mathcal{O}}

\def\cR{\mathcal{R}}

\def\cW{\mathcal{W}}

\def\sf0{{\mathsf{0}}}

\def\vec{\mathrm{vec}~}

